%% file: matchingreconfig.tex
\documentclass[a4paper]{article}

\usepackage{microtype}
\usepackage{latexsym,amsmath}
\usepackage{amsthm}
\usepackage{amssymb}
\usepackage{enumerate}
\usepackage{complexity}
\usepackage{xspace}
\usepackage{enumitem}
\usepackage{thm-restate}
\usepackage{authblk}
\usepackage{doi}

\usepackage{graphicx}
\usepackage{color}
\usepackage{tikz}
\usetikzlibrary{positioning,calc,decorations.pathmorphing,decorations.pathreplacing,patterns,shapes}
\usepackage{soul}
\usepackage[textsize=footnotesize,color=green!40]{todonotes}
\usepackage{todonotes}
\usepackage{ifthen}
\usepackage{mathtools}

\newtheorem{theorem}{Theorem}
\newtheorem{lemma}[theorem]{Lemma}
\newtheorem{proposition}[theorem]{Proposition}
\newtheorem{fact}[theorem]{Fact}
\newtheorem{corollary}[theorem]{Corollary}
\newtheorem{remark}{Remark}
\theoremstyle{definition}
\theoremstyle{remark}
\newtheorem{claim}{Claim}

\newtheorem{case}{Case}

\newcommand{\tj}{\textsc{Token Jumping Reconfiguration}\xspace}
\newcommand{\ts}{\textsc{Token Sliding Reconfiguration}\xspace}

\newcommand{\MR}{\textsc{Matching Reconfiguration}\xspace}
\newcommand{\MDIST}{\textsc{Matching Distance}\xspace}
\newcommand{\EMDIST}{\textsc{Exact Matching Distance}\xspace}
\newcommand{\EMDIAM}{\textsc{Exact Matching Diameter}\xspace}
\newcommand{\DST}{\textsc{Directed Steiner Tree}\xspace}

\newcommand{\setcover}{\textsc{Set Cover}\xspace}

\newcommand{\sets}{\ensuremath{\mathcal{S}}\xspace}
\newcommand{\opt}{\ensuremath{\operatorname{OPT}}\xspace}

\newcommand{\symdiff}{\ensuremath{\mathop{\bigtriangleup}}}
\newcommand*{\nat}{\ensuremath{\mathbb{N}}} 
\newcommand{\matchgraph}[1][k]{\ensuremath{\mathcal{M}_{#1}}}
\newcommand{\dist}{\ensuremath{\operatorname{dist}}}
\newcommand{\diam}{\ensuremath{\operatorname{diam}}}

\newcommand{\ini}{s}
\newcommand{\tar}{t}
\newcommand{\yes}{\textsc{Yes}\xspace}
\newcommand{\no}{\textsc{No}\xspace}

\title{Shortest Reconfiguration of Matchings}
\author[1]{Nicolas Bousquet}
\author[2]{Tatsuhiko Hatanaka}
\author[2]{Takehiro Ito}
\author[3]{Moritz M\"uhlenthaler}

\affil[1]{CNRS, Laboratoire G-SCOP, Grenoble-INP, Univ. Grenoble-Alpes, Grenoble, France}
\affil[2]{Graduate School of Information Sciences, Tohoku University, Japan}
\affil[3]{Fakult\"at f\"ur Mathematik, TU Dortmund University, Germany}

\begin{document}

\maketitle

\begin{abstract}
  Imagine that unlabelled tokens are placed on the edges of a graph, such that
  no two tokens are placed on incident edges. A token can jump to another
  edge if the edges having tokens
  remain independent. We study the problem of determining the \emph{distance}
  between two token configurations (resp., the corresponding matchings), which
  is given by the length of a shortest transformation. 
  We give a polynomial-time
  algorithm for the case that at least one of the two configurations
  is not inclusion-wise maximal and show that otherwise, the problem admits no
  polynomial-time sublogarithmic-factor approximation unless $\P =
  \NP$. Furthermore, we show that the distance of two configurations in
  bipartite graphs is fixed-parameter tractable parameterized by the size $d$
  of the symmetric difference of the source and target configurations, and
  obtain a $d^\varepsilon$-factor approximation algorithm  for every
  $\varepsilon > 0$ if additionally the configurations correspond to maximum
  matchings. Our two main technical tools are the Edmonds-Gallai decomposition
  and a close relation to the \textsc{Directed Steiner Tree}
  problem. 
  Using the former, we also characterize those graphs whose
  corresponding configuration graphs are connected. 
  Finally, we show that deciding if the distance between two configurations is
  equal to a given number $\ell$ is complete for the class $\D^\P$, and
  deciding if the diameter of the graph of configurations is equal to $\ell$ is
  $\D^\P$-hard.
\end{abstract}
\clearpage

\input{0_introduction}

\input{1_mdist_hardness}

\input{2_fpt}

\input{3_exact_distance}

\bibliographystyle{abbrv}
\bibliography{bibliomatching}

\newpage
\appendix
\input{appendix1}
\input{appendix2}

\input{appendix3}
\end{document}

%% file: 0_introduction.tex
\section{Introduction}

A reconfiguration problem asks for the existence of a step-by-step
transformation between two given configurations, where in each step we apply
some simple modification to the current configuration. The set of
configurations may for instance be the set of
$k$-colorings~\cite{BonamyB13,Feghali0P15} or independent
sets~\cite{Ito:11,KaminskiMM12,LM:18} of a graph, or the set of satisfying
assignments of Boolean formulas~\cite{Gopalan09}. A suitable modification may
for example alter the color of a single vertex, or the truth value of a
variable in a satisfying assignment.  For a survey on reconfiguration problems,
the reader is referred to~\cite{vHeuvel13} or~\cite{Nishimura17}.

Recently, there has been considerable interest in the complexity of finding
\emph{shortest} transformations between configurations. Examples include finding
a shortest transformation between triangulations of planar point
sets~\cite{Pilz:14} and simple polygons~\cite{Aichholzer:15}, configurations of
the Rubik's cube~\cite{DER:17}, and satisfying assignments of Boolean
formulas~\cite{Mouawad:17}. For all of these problems, except the last one, we can 
decide efficiently if a transformation between two given configurations exists.
However, deciding if there is a transformation of at most a given length is
\NP-complete. In particular, the flip distance of triangulations of planar
point sets is known to be \APX-hard~\cite{Pilz:14} and, on the positive side,
fixed-parameter tractable (FPT) in the length of the
transformation~\cite{Li:17}. Our reference problem is the task of computing the
length of a shortest transformation between matchings of a graph. We show that
even in a very restricted setting the problem admits no $o(\log n)$-factor
approximation  unless $\P = \NP$ and we give polynomial-time algorithms in some
special cases. Furthermore, we show that the problem is FPT in the size of the
symmetric difference of the two given configurations, which implies that it is
also FPT in the length of the transformation.

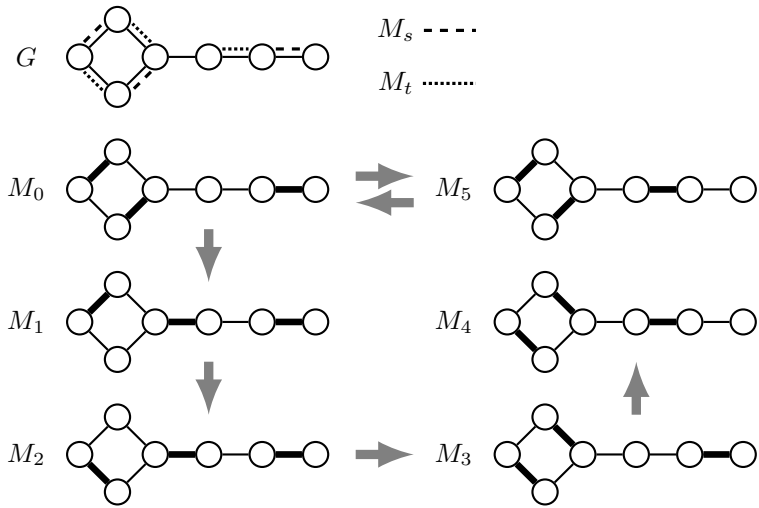
\begin{figure}[t]
	\centering
	\begin{tikzpicture}[vertex/.style={shape=circle,thick,draw,node distance=2em},copyv/.style={shape=circle,thick,draw,node distance=2em},medge/.style={very thick,dashed},nedge/.style={very thick,densely dotted},token/.style={line width=2.5pt},arrow/.style={-latex,line width=3.5pt,color=gray}]
	\node[vertex]			(c1)	{};
	\node[vertex,below left of=c1]	(c2)	{};
	\node[vertex,above left of=c2]	(c3)	{};
	\node[vertex,above right of=c3]	(c4)	{};
	
	\node[vertex,right of=c1] (p1) {};
	\node[vertex,right of=p1] (p2) {};
	\node[vertex,right of=p2] (p3) {};
	
	\node (G) at ($(c3)+(-2em,0)$) {$G$};
	\node (M) at ($(p3)+(3em,1em)$) {$M_\ini$};
	\node (N) at ($(p3)+(3em,-1em)$) {$M_\tar$};
	\draw[medge] (M) -- ++(3em,0);
	\draw[nedge] (N) -- ++(3em,0);
	
	\foreach \s in {1,...,4} {
		\pgfmathparse{int(Mod(\s,4)+1)}
		\edef\nxt{\pgfmathresult}
		\draw[thick] (c\s) -- (c\nxt);
	}
	\draw[thick] (c1) -- (p1) -- (p2) -- (p3);
	
	\draw[medge] ([xshift=2pt,yshift=-2pt]c1.south west) -- ([xshift=2pt,yshift=-2pt]c2.north east);
	\draw[medge] ([xshift=-2pt,yshift=2pt]c3.north east) -- ([xshift=-2pt,yshift=2pt]c4.south west);
	\draw[nedge] ([xshift=2pt,yshift=2pt]c1.north west)  -- ([xshift=2pt,yshift=2pt]c4.south east);
	\draw[nedge] ([xshift=-2pt,yshift=-2pt]c2.north west) -- ([xshift=-2pt,yshift=-2pt]c3.south east);
	\draw[medge] ([yshift=3pt]p2.east) -- ([yshift=3pt]p3.west);
	\draw[nedge] ([yshift=3pt]p1.east) -- ([yshift=3pt]p2.west);
	
	\foreach \t in {0,1,2}{
		\node[copyv]	at ($(c1)+(0,-5em)+5*(0,-\t em)$) (c\t1) {};
		\node[copyv,below left of=c\t1]	(c\t2)	{};
		\node[copyv,above left of=c\t2]	(c\t3)	{};
		\node[copyv,above right of=c\t3]	(c\t4)	{};
		
		\node[copyv,right of=c\t1] (p\t1) {};
		\node[copyv,right of=p\t1] (p\t2) {};
		\node[copyv,right of=p\t2] (p\t3) {};
		
		\node (M\t) at ($(c\t3)+(-2em,0)$) {$M_{\t}$};
		
		\foreach \s in {1,...,4} {
			\pgfmathparse{int(Mod(\s,4)+1)}
			\edef\nxt{\pgfmathresult}
			\draw[thick] (c\t\s) -- (c\t\nxt);
		}
		\draw[thick] (c\t1) -- (p\t1) -- (p\t2) -- (p\t3);
	}
	\foreach \t in {3,4,5}{
		\node[copyv]	at ($(p23)+(0,-15em)+5*(2em,\t em)$) (c\t1) {};
		\node[copyv,below left of=c\t1]	(c\t2)	{};
		\node[copyv,above left of=c\t2]	(c\t3)	{};
		\node[copyv,above right of=c\t3]	(c\t4)	{};
		
		\node[copyv,right of=c\t1] (p\t1) {};
		\node[copyv,right of=p\t1] (p\t2) {};
		\node[copyv,right of=p\t2] (p\t3) {};
		
		\node (M\t) at ($(c\t3)+(-2em,0)$) {$M_{\t}$};
		
		\foreach \s in {1,...,4} {
			\pgfmathparse{int(Mod(\s,4)+1)}
			\edef\nxt{\pgfmathresult}
			\draw[thick] (c\t\s) -- (c\t\nxt);
		}
		\draw[thick] (c\t1) -- (p\t1) -- (p\t2) -- (p\t3);
	}

	\draw[token] (c01) -- (c02);
	\draw[token] (c03) -- (c04);
	\draw[token] (p02) -- (p03);
	\draw[token] (c11) -- (p11);
	\draw[token] (c13) -- (c14);
	\draw[token] (p12) -- (p13);
	\draw[token] (c21) -- (p21);
	\draw[token] (c23) -- (c22);
	\draw[token] (p22) -- (p23);
	\draw[token] (c31) -- (c34);
	\draw[token] (c33) -- (c32);
	\draw[token] (p32) -- (p33);
	\draw[token] (c41) -- (c44);
	\draw[token] (c43) -- (c42);
	\draw[token] (p42) -- (p41);
	\draw[token] (c51) -- (c52);
	\draw[token] (c53) -- (c54);
	\draw[token] (p52) -- (p51);
	
	\draw[arrow] ($(p01)+(-90:1.5em)$) -- ($(p11)+(90:1.5em)$);
	\draw[arrow] ($(p11)+(-90:1.5em)$) -- ($(p21)+(90:1.5em)$);
	\draw[arrow] ($(p23)+(0:1.5em)$) -- ($(M3)+(180:1.5em)$);
	\draw[arrow] ($(p31)+(90:1.5em)$) -- ($(p41)+(-90:1.5em)$);
	\draw[arrow] ($(p03)+(0:1.5em)+(0,5pt)$) -- ($(M5)+(180:1.5em)+(0,5pt)$);
	\draw[arrow] ($(M5)+(180:1.5em)+(0,-5pt)$) -- ($(p03)+(0:1.5em)+(0,-5pt)$);
	
	\end{tikzpicture}
	\caption{A reconfiguration sequence $M_\ini = M_0, M_1, \ldots, M_4 = M_\tar$ of matchings in a graph $G$.}
	\vspace{-1em}
	\label{fig:example}
\end{figure}

\paragraph*{Reconfiguration of matchings.}

Deciding if there is a transformation between two matchings of a graph is known
as an early example of a reconfiguration problem that admits a non-trivial 
polynomial-time algorithm~\cite{Ito:11}. Recall that a \emph{matching} $M$ of
a graph is a set
of pairwise independent edges. (\figurename~\ref{fig:example} shows the six
different matchings of the graph $G$.)
We may consider a matching as a placement of (unlabeled) \emph{tokens} on independent edges: 
Then, the \emph{Token Jumping} (\emph{TJ}) operation provides an
adjacency relation on the set of matchings of a graph, all having the same 
cardinality\footnote{There is another well-studied operation,
called Token Sliding (TS), for reconfiguration of subgraphs having the same
cardinality. In this paper, we employ TJ as the default operation. However,
some of our results apply also to TS, because TJ and TS are equivalent for
maximum-cardinality matchings.}: 
Two matchings $M$ and $M'$ of a graph
$G$ are \emph{adjacent} (under TJ) if one can be obtained from the other by
relocating a single token, that is, if $|M \setminus M'| = 1$ and $|M'\setminus
M|=1$.
We say that a sequence $M_0,M_1,\ldots,M_\ell$ of matchings of $G$ is a
\emph{reconfiguration sequence} of \emph{length} $\ell$ from $M$ to $M'$, if
$M_0 = M$, $M_\ell = M'$, and $M_{i-1}$ and $M_i$ are adjacent for each $i$, $1
\leq i \leq \ell$.  (See the sequence $M_0, M_1, \ldots, M_4$ in
\figurename~\ref{fig:example} as an example.) The following question is often
referred to as the \emph{reachability variant} of the matching reconfiguration
problem:
\begin{quote}
  \MR \\
  \textbf{Input:} Graph $G$ and two matchings $M_\ini,M_\tar$ of $G$.\\
  \textbf{Question:} Is there a reconfiguration sequence from $M_\ini$ to $M_\tar$?
\end{quote}

For \yes-instances of the above problem the polynomial-time algorithm given
in~\cite{Ito:11} gives a bound of $O(n^2)$ on the length of a transformation.
The \emph{distance} between two matchings is the length of a shortest
transformation between them (under TJ).  If there is no transformation between
two matchings, we regard their distance as infinity. In this paper we study
the complexity of the following optimization problem related to matching
reconfiguration, which is also referred to as the \emph{shortest variant}. 
\begin{quote}
  \MDIST \\
  \textbf{Input:} Graph $G$ and two matchings $M_\ini,M_\tar$ of $G$.\\
  \textbf{Task:} Compute the distance between $M_\ini$ and $M_\tar$.
\end{quote}

We also study two related \emph{exact} problems.  The first is the exact
version of \MDIST, which takes as input also the supposed distance $\ell$ of
the given matchings.
\begin{quote}
  \EMDIST \\
  \textbf{Input:} Graph $G$, matchings $M_\ini,M_\tar$ of $G$, and number $\ell \in \nat$.\\
  \textbf{Question:} Is $\ell$ equal to the distance between $M_\ini$ and $M_\tar$?
\end{quote}
The second decides the maximum distance (\emph{diameter}) of any two matchings of a given cardinality $k$ in a graph. 
\begin{quote}
  \EMDIAM \\
  \textbf{Input:} Graph $G$ and numbers $k, \ell \in \nat$.\\
  \textbf{Question:} Is $\ell$ equal to the maximum distance between any two matchings of cardinality $k$ of $G$?
\end{quote}

\paragraph*{Related results.}

Despite recent intensive studies on reconfiguration problems (see, e.g., a
survey~\cite{Nishimura17}), most of known algorithmic (positive) results are
obtained for reachability variants.  However, they sometimes give answers to
their shortest variants: if the algorithm constructs an actual reconfiguration
sequence which, at any step, transforms an edge of the initial matching into an
edge of the target one, then the sequence is indeed a shortest one.

Generally speaking, finding shortest transformations is much more difficult if
we need a \emph{detour}, which touches an element that is not in the symmetric
difference of the source and target configurations.
For such a detour-required case, only a few polynomial-time algorithms are
known for shortest variants, e.g., satisfying assignments of a certain Boolean
formulas by Mouawad et al.~\cite{MouawadNPR17}, and independent sets under the
TS operation for caterpillars by Yamada and Uehara~\cite{YamadaU16}.  Note that
\MR belongs to the detour-required case; recall the example in
\figurename~\ref{fig:example}, where we need to use the edge in $E(G) \setminus
(M_\ini \cup M_\tar)$ in any reconfiguration sequence.
 
The reconfiguration of matchings is a special case of the reconfiguration of
independent sets of a graph.  To see this, recall that matchings of a graph
correspond to independent sets of its line graph. 
Therefore, by a result of Kami\'nski et al.~\cite{KaminskiMM12}, we can solve
\MDIST in polynomial time if the line graph of a given graph is even-hole free.
Note that in this case no detour is required.

\paragraph*{Our results.}
Although the reconfiguration of independent sets is one of the most
well-studied reconfiguration problems~(see, e.g., a survey~\cite{Nishimura17}),
to the best of our knowledge, the shortest variant of independent sets under
the TJ operation is known to be solvable only for even-hole-free graphs, as
mentioned above.  Thus, in this paper, we start a systematic study of the
complexity of finding shortest reconfiguration sequences between matchings, and
more generally, between independent sets of a graph.   

Our first result is the following classification of the complexity of the
problem \MDIST. It follows immediately from
Theorem~\ref{thm:inapprox:msr}, Corollary~\ref{cor:inapprox:nonmaximum}, and
Lemma~\ref{lemma:ptime:msr}.

\begin{restatable}{theorem}{thm:main}
  \MDIST can be solved in polynomial time if at least
  one of the two matchings is not inclusion-wise maximal. Furthermore, \MDIST
  restricted to instances where both matchings are maximal admits no
  polynomial-time $o(\log n)$-factor approximation algorithm, unless $\P =
  \NP$. 
  \label{thm:main}
\end{restatable}

\noindent 
The hardness part of Theorem~\ref{thm:main} holds even for bipartite graphs of
maximum degree three. Note that it implies approximation hardness for shortest
transformations between $b$-matchings of a graph and for shortest
transformations between independent sets on any graph class containing line
graphs.

On the positive side, we show that determining the distance of maximum
matchings of \emph{bipartite} graphs is FPT in the size $d$ of the symmetric difference of
the input matchings. In our algorithm we consider two cases: either a shortest
reconfiguration sequence contains a non-inclusion-wise maximal matching or
not. Extending the positive side of Theorem~\ref{thm:main}, we give a
polynomial-time algorithm for the former case. To deal with the latter case, we
proceed in two stages. We first generate (many) instances of \MDIST, such that
the two input matchings are \emph{maximum}. This allows us to make some
additionial assumptions based on the Edmonds-Gallai
decomposition~\cite[Ch.~24.4b]{Schrijver:CO}. We then further reduce this variant
of \MDIST to \DST with at most $d/2$ terminals and show that optimal Steiner
trees are in correspondence with shortest reconfiguration sequences. Optimal
Steiner trees can be computed in FPT time for a constant number of terminals
according to the algorithms~\cite{Bjorklund:07,DW:71}. By properly combining
the reconfiguration sequences obtained from optimal Steiner trees we obtain the
following result.

\begin{restatable}{theorem}{fptresult}
  \MDIST in bipartite graphs can be solved in time $2^d \cdot n^{O(1)}$, where
  $d$ is the size of the symmetric difference of two given matchings.
  \label{thm:fpt} 
\end{restatable}

This result raises hopes for possible generalizations, e.g., an FPT algorithm
for finding a shortest transformation between independent sets of claw-free
graphs.  The reduction from \MDIST restricted to maximum matchings in bipartite
graphs to \DST is approximation-preserving, which implies the following.

\begin{corollary}
  \MDIST restricted to maximum matchings in bipartite graphs admits a
  polynomial-time $d^{\varepsilon}$-factor approximation algorithm for every
  $\varepsilon > 0$, where $d$ is the size of the symmetric difference of two
  given matchings.
  \label{cor:approximation}
\end{corollary}

We complement Theorem~\ref{thm:main} by showing that
there is a polynomial-time algorithm that decides if the maximum distance
between any two matchings of a graph is \emph{finite}. However, 
we also show that the problems \EMDIST and \EMDIAM are both hard for
the class $\D^\P$, which contains \NP\xspace and \coNP.

\begin{restatable}{theorem}{dpcompleteness}
  The problem \EMDIST is complete for $\D^\P$ and \EMDIAM is $\D^\P$-hard.
  \label{thm:dpc}
\end{restatable}

The class $\D^\P$ is a class of decision problems introduced by Papadimitriou and Yannakakis in~\cite{PY:84}. 
It is a natural class for \emph{exact} problems,
\emph{critical} problems, and for example for the question, whether a certain
inequality is a facet of a polytope~\cite{PY:84}. It was proved by Frieze and
Teng that the related problem of deciding the diameter of the graph of a
polyhedron is also $\D^\P$-hard~\cite{FT:94}. 

\paragraph*{Notation.}
We denote by $A \symdiff B$, the symmetric difference of two sets $A$ and $B$.
That is, $A \symdiff B := (A \setminus B) \cup (B \setminus A)$.
Unless stated otherwise, graphs are simple. For standard
definitions and notation related to graphs, we refer the reader to~\cite{Diestel2005}. 
For a graph $G$, we denote by $V(G)$ (resp., $E(G)$) the set of vertices
(resp., edges) of $G$.  We denote by $\overline G$ the complement graph of $G$
and by $\overline E$ the edge-set $E(\overline{G})$ of the complement graph.
Let $G = (V, E)$ be a graph and let $M \subseteq E$ be a matching.
A vertex of $G$ that is not incident to any edge in $M$ is called
\emph{$M$-exposed} or \emph{$M$-free}, otherwise it is \emph{matched} or \emph{covered}.  
It will sometimes be convenient to work with the \emph{reconfiguration graph}
$\matchgraph(G)$ of matchings of a graph $G$, which is defined as follows.
\begin{align*}
  V(\matchgraph(G)) &:= \{ M \subseteq E \mid \text{ $M$ is a matching in $G$, $|M| = k$} \}\\
  E(\matchgraph(G)) &:= \{ MN \mid M, N \in V(\matchgraph(G)),\, |M \symdiff N| = 2\}
\end{align*}
We denote by $\dist_{\matchgraph(G)}(M,N)$ the distance of two matchings in
$\matchgraph(G)$ and by $\diam(\matchgraph(G))$ the maximal distance of any two
vertices of $\matchgraph(G)$. 

The remainder of this paper is organized as follows. In
Section~\ref{sec:inapprox}, we prove Theorem~\ref{thm:main}. In
Section~\ref{sec:FPT}, we provide the FPT algorithm and the approximation
algorithm for finding a shortest reconfiguration sequence for two maximum
matchings of a bipartite graph. Section~\ref{sec:diameter} contains our
hardness results for deciding the exact distance of two matchings and the
exact diameter of the graph of matchings.

%% file: 1_mdist_hardness.tex
\section{Hardness of {\MDIST}}
\label{sec:inapprox}

The goal of this section is to prove Theorem~\ref{thm:main}. The positive part
of~\ref{thm:main} is a consequence of the following lemma. 

\begin{lemma}[$\ast$]
  \MDIST restricted to instances where at least one of the two matchings is
  not inclusion-wise maximal can be decided in polynomial time.
  \label{lemma:ptime:msr}
\end{lemma}

To prove Lemma~\ref{lemma:ptime:msr}, we show that the distance of two
matchings $M_s$ and $M_t$, at least one of which is not inclusion-wise maximal,
is either $|M_s \symdiff M_t| / 2$ or $|M_s \symdiff M_t|/2 + 1$. Furthermore,
it can be checked in polynomial time, which case applies.
To prove the hardness part of Theorem~\ref{thm:main}, we show the following.

\begin{theorem}
   \MDIST admits no $o(\log n)$-factor approximation unless $\P=\NP$, even when
   restricted to instances on bipartite graphs of maximum degree three.
    \label{thm:inapprox:msr}
\end{theorem}

To show approximation hardness we prove in
Section~\ref{sec:inapprox:hardness}, that a sublogarithmic-factor approximation
for \MDIST yields a sublogarithmic-factor approximation of \setcover, using the
construction from Section~\ref{sec:inapprox:construction}. However, the
\setcover problem is not approximable within a sublogarithmic factor, unless
$\P = \NP$~\cite{DS:14}.

\begin{remark}
    The hardness part of Theorem~\ref{thm:main} also holds for the
    Token Sliding operation~\cite{KaminskiMM12} since $M_1$ and $M_2$ are maximum.
  \label{rem:max}
\end{remark}

Let us briefly recall some definitions related to the \setcover problem. An
instance $I = (U, \sets)$ of \setcover is given by a set $U$ called
\emph{items} and a family $\sets$ of subsets of $U$ called \emph{hyperedges}.
The task is to find the minimal number of sets in \sets that are required to
cover $U$. We denote this number by $\opt(I)$ and let $n := |U|$ and $m :=
|\sets|$.  
Let $d :=
\max_{S \in \sets} \{ |S| \}$ be the maximum cardinality of a set in \sets and
for each $u \in U$ let $f_u = |\{S \in \sets \mid u \in S\}|$ be the
\emph{frequency} of $u$. Furthermore, let $f := \max_{u \in U} \{f_u\}$ be
frequency of $I$.

\subsection{The Construction}
\label{sec:inapprox:construction}

We construct from the \setcover instance $I = (U, \sets)$ an instance $I' = (G,
M_1, M_2)$ of \MDIST. An illustration of the construction shown in Figure~\ref{fig:inapprox:example}.

First, for each item $u \in U$ we create a 4-cycle $C_u$ on the vertices
$c_u^1, c_u^2, c_u^3, c_u^4$ and a path $P_u$ on the vertices $p_u^1,
p_u^2, \ldots, p_u^{2f_u}$. We connect $c_u^1$ and $p_u^1$ with an edge for each $u
\in U$.  Furthermore, for each $S \in \sets$, we create a path $P_S$ on the
vertices $p_S^1, p_S^2, \ldots, p_S^{2|S|}$ and a path $Q_S$ on the vertices
$q_S^1, q_S^2, \ldots, q_S^L$, where $L$ is an \emph{odd} number that will be
specified later. For each $S \in \sets$, we connect $P_S$ to $Q_S$ with the
edge $p_S^{2|S|}q_S^1$.  
Let us now simulate the containment in the hyperedges as follows.  For each
item $u \in U$, the \emph{terminals} of $P_u$ are the vertices $p_u^i$ where
$i$ is even.  For each hyperedge $S \in \sets$, the \emph{terminals} of $P_S$
are the vertices $p_S^i$ where $i$ is odd.  Now, we add edges forming a matching on the
terminal vertices such that there is a (unique) edge in the matching between
the terminals of $P_u$ and the terminals of~$P_S$ if and only if $u \in S$.
Note that such a matching exists since there are $|S|$ terminals in~$P_S$ and
$f_u$ terminals in $P_u$. 
Note that $G$ is bipartite, since we create edges between vertices 
with an even index and vertices with an odd index. Also note that since we assume that
$m=\poly(n)$, the size of the instance is polynomial in $n + m$, whenever $L$ is.

It remains to construct two matchings $M_1$ and $M_2$ of $G$. Observe that
$P_u$ and $P_S$ (for $u \in U$ and $S \in \sets$) are paths with an even number
of vertices, they admit a unique perfect matching $M_u$ ($M_S$) for each
$u \in U$ ($S \in \sets$). Similarly, the path $Q_S$ admits a matching 
$N_S := \{q_S^1q_S^2,q_S^3q_S^4, \ldots q_S^{L-2}q_S^{L-1}\}$. Note that
$Q_S$ also admits another perfect matching if we push the matching to the right
(this will be of importance for the proof of Theorem~\ref{thm:inapprox:msr}).
Now let $M' = \bigcup_{u \in U} M_u \cup \bigcup_{S \in \sets} M_S \cup N_S$.
The matching $M_1$ is $M'$ plus edges $c_u^1c_u^2$
and~$c_u^3c_u^4$ for each~$u \in U$. And the matching $M_2$ is $M'$
plus the edges $E(C_u) \setminus M_1 = \{ c_u^2c_u^3,
c_u^1c_u^4\}$ for each~$u \in U$.  That is, $M_1$ and $M_2$ only differ
on the 4-cycles in $G$ and the vertices~$q_S^L$ are $M_1$-free and $M_2$-free for each $S
\in \sets$. This completes the construction of the instance~$I' = (G, M_1,
M_2)$. Note that the matchings $M_1$ and $M_2$ are maximum.

Figure~\ref{fig:inapprox:example} illustrates the above construction. The dashed edges are $M_1$-edges and the
dotted edges are $M_2$ edges. The terminal nodes that are used in order to
model the incidence relation of the \setcover instance are indicated by the
square shapes. Informally, we think of $L$ as a very large number, so in order
to find a short reconfiguration sequence from $M_1$ to~$M_2$, it is desirable to
minimize the number of times an alternating path from $q_S^1$ to $q_S^L$ is
reconfigured in order to switch the matching edges on each cycle gadget.

\begin{figure}[t]
  \centering
  \begin{tikzpicture}[vertex/.style={shape=circle,thick,draw,node distance=2.5em},terminal/.style={shape=rectangle,thick,draw,node distance=2.5em},medge/.style={very thick,dashed},nedge/.style={very thick,densely dotted}]
    \node[vertex]			(c11)	{};
    \node[vertex,below left of=c11]	(c12)	{};
    \node[vertex,above left of=c12,label=left:$C_1$]	(c13)	{};
    \node[vertex,above right of=c13]	(c14)	{};

    \node[vertex,below=5em of c11]	(c21)	{};
    \node[vertex,below left of=c21]	(c22)	{};
    \node[vertex,above left of=c22,label=left:$C_2$]	(c23)	{};
    \node[vertex,above right of=c23]	(c24)	{};

    \node[vertex,below=5em of c21]	(c31)	{};
    \node[vertex,below left of=c31]	(c32)	{};
    \node[vertex,above left of=c32,label=left:$C_3$]	(c33)	{};
    \node[vertex,above right of=c33]	(c34)	{};

    \node[vertex,right of=c11,label=below:{$p_1^1$}] (p11) {};
    \node[vertex,right of=c21,label=below:{$p_2^1$}] (p21) {};
    \node[vertex,right of=c31,label=below:{$p_3^1$}] (p31) {};
    \node[terminal,right of=p31,label=below:{$p_3^2$}] (p32) {};

    \draw[thick] (p31) -- (p32);
    \path[draw,thick] (c11) to [bend left=0] (p11);
    \path[draw,thick] (c21) to [bend left=0] (p21);
    \path[draw,thick] (c31) to [bend left=0] (p31);

    \foreach \s in {1,...,4} {
      \pgfmathparse{int(Mod(\s,4)+1)}
      \edef\nxt{\pgfmathresult}
      \draw[thick] (c1\s) -- (c1\nxt);
      \draw[thick] (c2\s) -- (c2\nxt);
      \draw[thick] (c3\s) -- (c3\nxt);
    }

    \foreach \s in {2,...,4} {
      \pgfmathparse{int(\s-1)}
      \edef\prv{\pgfmathresult}
      \pgfmathparse{int(Mod(\s-1,2))}
      \edef\party{\pgfmathresult}
      \ifthenelse{\party=0}
      {\node[vertex,right of=p1\prv,label=below:$p_1^{\s}$] (p1\s) {};}
      {\node[terminal,right of=p1\prv,label=below:$p_1^{\s}$] (p1\s) {};}
      
      \ifthenelse{\party=0}
      {\node[vertex,right of=p2\prv,label=below:$p_2^{\s}$] (p2\s) {};}
      {\node[terminal,right of=p2\prv,label=below:$p_2^{\s}$] (p2\s) {};}

      \draw[thick] (p1\prv) -- (p1\s);
      \draw[thick] (p2\prv) -- (p2\s);
    }

    \node[terminal,right of=p14,label=below:$p_{S_1}^1$] (P11) {};
    \node[vertex,right of=P11,label=below:$p_{S_1}^2$] (P12) {};
    \node[terminal,right of=p24,label=below:$p_{S_2}^1$] (P21) {};
    \node[terminal,label=below:$p_{S_3}^1$] (P31) at (P21.center|-p32.center)  {};

    \foreach \s in {2,...,4} {
      \pgfmathparse{int(\s-1)}
      \edef\prv{\pgfmathresult}
      \pgfmathparse{int(Mod(\s,2))}
      \edef\party{\pgfmathresult}
      \ifthenelse{\party=0}
      {
          \node[vertex,right of=P2\prv,label=below:$p_{S_2}^\s$] (P2\s) {};
          \node[vertex,right of=P3\prv,label=below:$p_{S_3}^\s$] (P3\s) {};
      }
      {
          \node[terminal,right of=P2\prv,label=below:$p_{S_2}^\s$] (P2\s) {};
          \node[terminal,right of=P3\prv,label=below:$p_{S_3}^\s$] (P3\s) {};
      }
      \draw[thick] (P2\prv) -- (P2\s);
      \draw[thick] (P3\prv) -- (P3\s);
    }
    
    \node[vertex,right of=P24,label=below:$q_{S_2}^1$] (Q21) {};
    \node[vertex,right of=P34,label=below:$q_{S_3}^1$] (Q31) {};
    \node[vertex,label=below:$q_{S_1}^1$] (Q11) at (Q21.center|-P12.center) {};
    \foreach \s in {1,2,3} {
      \node[vertex,right of=Q\s1,label=below:$q_{S_\s}^2$] (Q\s2) {};
      \node[right of=Q\s2] (dots\s) {\ldots};
      \node[vertex,right of=dots\s,label=below:$q_{S_\s}^L$] (QL\s) {};
      \draw[thick] (Q\s1) -- (Q\s2);
      \draw[thick] (Q\s2) -- (dots\s);
      \draw[thick] (dots\s) -- (QL\s);
    }

    \path[draw,thick] (P12) to [bend left=0] (Q11);
    \path[draw,thick] (P24) to [bend left=0] (Q21);
    \path[draw,thick] (P34) to [bend left=0] (Q31);

    \draw[thick] (p14) -- (P11);
    \path[thick,draw] (p12) to [bend right=15] (P23);
    \draw[thick] (p24) -- (P21);
    \path[thick,draw] (p22) to [bend right=15] (P33);
    \path[thick,draw] (p32) to (P31);

    \draw[very thick, decorate,decoration={brace,amplitude=1em}] ([yshift=2.5em]p11.north west) -- ([yshift=2.5em]p14.north east) node [midway,yshift=2em] {$P_u$};
    \draw[very thick, decorate,decoration={brace,amplitude=1em}] ([yshift=2.5em]P11.north west) -- ([yshift=2.5em]P24.north east|-P12.north east) node [midway,yshift=2em] {$P_S$};
    \draw[very thick, decorate,decoration={brace,amplitude=1em}] ([yshift=2.5em]Q11.north west) -- ([yshift=2.5em]QL1.north east) node [midway,yshift=2em] {$Q_S$};
    \draw[very thick, decorate,decoration={brace,amplitude=1em}] ([yshift=2.5em]c13.north west) -- ([yshift=2.5em]c11.north east) node [midway,yshift=2em] {$C_u$};

    \foreach \s in {1,3} {
        \pgfmathparse{int(\s+1)}
        \edef\nxt{\pgfmathresult}
        \draw[medge] ([yshift=3pt]p1\s.east) -- ([yshift=3pt]p1\nxt.west);
        \draw[medge] ([yshift=3pt]p2\s.east) -- ([yshift=3pt]p2\nxt.west);
        \draw[medge] ([yshift=3pt]P2\s.east) -- ([yshift=3pt]P2\nxt.west);
        \draw[medge] ([yshift=3pt]P3\s.east) -- ([yshift=3pt]P3\nxt.west);

        \draw[nedge] ([yshift=-3pt]p1\s.east) -- ([yshift=-3pt]p1\nxt.west);
        \draw[nedge] ([yshift=-3pt]p2\s.east) -- ([yshift=-3pt]p2\nxt.west);
        \draw[nedge] ([yshift=-3pt]P2\s.east) -- ([yshift=-3pt]P2\nxt.west);
        \draw[nedge] ([yshift=-3pt]P3\s.east) -- ([yshift=-3pt]P3\nxt.west);
    }

    \draw[medge] ([yshift=3pt]P11.east) -- ([yshift=3pt]P12.west);
    \draw[nedge] ([yshift=-3pt]P11.east) -- ([yshift=-3pt]P12.west);
    \draw[medge] ([yshift=3pt]p31.east) -- ([yshift=3pt]p32.west);
    \draw[medge] ([yshift=3pt]Q11.east) -- ([yshift=3pt]Q12.west);
    \draw[medge] ([yshift=3pt]Q21.east) -- ([yshift=3pt]Q22.west);
    \draw[medge] ([yshift=3pt]Q31.east) -- ([yshift=3pt]Q32.west);

    \draw[nedge] ([yshift=-3pt]p31.east) -- ([yshift=-3pt]p32.west);
    \draw[nedge] ([yshift=-3pt]Q11.east) -- ([yshift=-3pt]Q12.west);
    \draw[nedge] ([yshift=-3pt]Q21.east) -- ([yshift=-3pt]Q22.west);
    \draw[nedge] ([yshift=-3pt]Q31.east) -- ([yshift=-3pt]Q32.west);

    \foreach \s in {1, 2, 3} {
        \draw[medge] ([xshift=2pt,yshift=-2pt]c\s1.south west) -- ([xshift=2pt,yshift=-2pt]c\s2.north east);
        \draw[medge] ([xshift=-2pt,yshift=2pt]c\s3.north east) -- ([xshift=-2pt,yshift=2pt]c\s4.south west);
        \draw[nedge] ([xshift=2pt,yshift=2pt]c\s1.north west)  -- ([xshift=2pt,yshift=2pt]c\s4.south east);
        \draw[nedge] ([xshift=-2pt,yshift=-2pt]c\s2.north west) -- ([xshift=-2pt,yshift=-2pt]c\s3.south east);
    }
  \end{tikzpicture}
  \caption{An example of the construction of an instance $(G, M_1, M_2)$ of
    \MDIST from an instance $(U, \sets)$ of \setcover, where $U = \{1, 2, 3\}$
    and $\sets = \{ S_1, S_2, S_3 \}$, $S_1 = \{1\}$, $S_2 = \{1, 2\}$, $S_3 = \{2,3\}$. The terminals are indicated by square node shapes.\label{fig:inapprox:example}}
\end{figure}
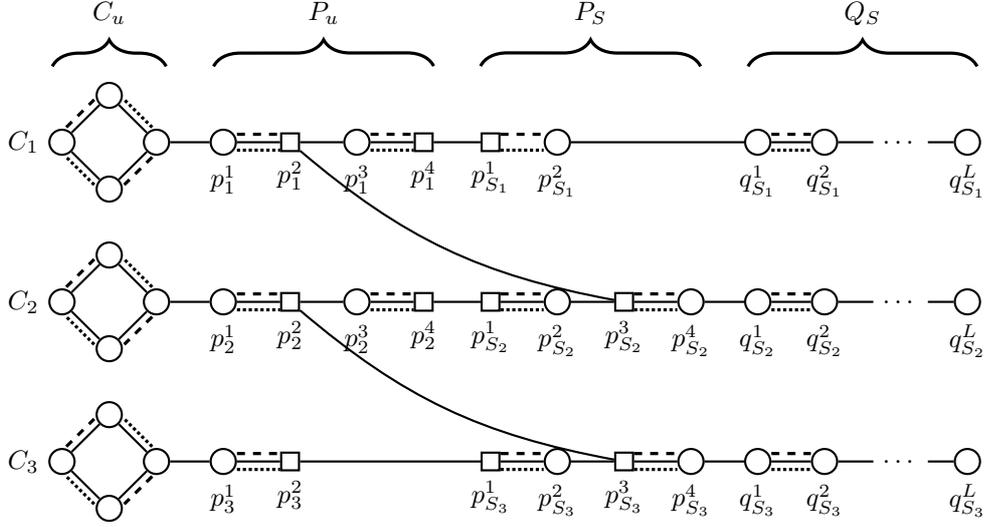

\subsection{Proof of Theorem~\ref{thm:inapprox:msr}}
\label{sec:inapprox:hardness}

Let $I = (U, \sets)$ be an instance of 
\setcover and let $I' = (G, M_1, M_2)$ is an instance of \MDIST, which is constructed
from $I$ as described in Section~\ref{sec:inapprox:construction}.
In order to prove the hardness-of-approximation result, we need to construct a
reconfiguration sequence from $M_1$ to $M_2$ from a cover and vice versa. 
We need two lemmas and let $L := |U|(2+f+d)$. Observe that since $L$ is polynomial
in $|I|$, the instance $I'$ can be constructed in polynomial time.

\begin{lemma}[$\ast$]
  Let $C \subseteq \sets$ be a cover of $U$. Then there is a reconfiguration
  sequence from $M_1$ to $M_2$ of length at most $2L|C| + 2|U| (2+f+d)$.
  \label{lemma:inapprox:ub}
\end{lemma}

\begin{lemma}[$\ast$]
  There is a polynomial-time algorithm $A'$ that constructs from a
  reconfiguration sequence $\tau$ from $M_1$ to $M_2$ of length $|\tau|$ a set
  cover $C \subseteq \sets$ of cardinality at most $|\tau| / 2L$.  
  \label{lemma:inapprox:lb}
\end{lemma}

We are now ready to prove Theorem~\ref{thm:inapprox:msr}.

\begin{proof}[Proof of Theorem~\ref{thm:inapprox:msr} (sketch)]
  Suppose we have an $f(n')$-factor approximation  algorithm $A$ for \MDIST, where
  $n' = |I'|$.  Then
  \begin{align*}
    A'(I) & \stackrel{Lemma~\ref{lemma:inapprox:lb}}{\leq}  \frac{A(I')}{2L}  \stackrel{Lemma~\ref{lemma:inapprox:ub}}{\leq} \frac{f(n')(\opt(I) \cdot 2L + 2|U| (f+d+2))}{2L} \leq 2f(n') \opt(I)\enspace.
  \end{align*}
  That is, we can compute in polynomial time a $2f(n')$-approximate solution of
  $I$. Combining this with the result of Dinur and Steurer~\cite{DS:14} completes
  the proof.
\end{proof}

With a slight modification of the proof, we can show that
the problem \MDIST remains hard, even if
the matchings are not maximum.

\begin{corollary}[$\ast$]
    \MDIST on bipartite graphs of maximum degree three
    and matchings that are not maximum does not admit a polynomial-time
    $o(\log n)$-factor approximation algorithm, unless $\P = \NP$.
    \label{cor:inapprox:nonmaximum}
\end{corollary}

%% file: 2_fpt.tex
\section{Distance of Matchings in Bipartite Graphs is FPT}
\label{sec:FPT}

The main result of this section is  Theorem~\ref{thm:fpt}, which states that
\MDIST for matchings on bipartite graphs is FPT, where the parameter is the
size of the symmetric difference of the source and target matchings. In the
following, let $(G, M_s, M_t)$ be an instance of \MDIST, where the graph $G =
(V, E)$ is bipartite.  Due to Lemma~\ref{lemma:ptime:msr}, we may assume that
$M_s$ and $M_t$ are inclusion-wise maximal, since otherwise we may find a
shortest transformation in polynomial time. Our FPT algorithm distinguishes two
cases. First, it may happen that in a shortest transformation some intermediate
matching is not inclusion-wise maximal. 
We show that in this case a shortest transformation may be found in polynomial
time. 

\begin{lemma}[$\ast$]
    There is a polynomial-time algorithm that outputs a shortest
    transformation from $M_s$ to $M_t$ via a matching $M$ that is not
    inclusion-wise maximal, or indicates that no such transformation exists.
    \label{lemma:fpt:non-iwm}
\end{lemma}

Note that Lemma~\ref{lemma:fpt:non-iwm} properly generalizes
Lemma~\ref{lemma:ptime:msr} and therefore gives a positive result beyond the
the complexity classification stated in Theorem~\ref{thm:main}. The proof idea is to find a
cheapest transformation into a matching that is not inclusion-wise maximal
respect to a certain cost measure that reflects the ``progress'' we make by
performing exchanges along a given augmenting path. The progress is essentially
the sum of the length of an augmenting path and the length of the remaining
transformation, which can be determined by Lemma~\ref{lemma:ptime:msr}. 

On the other hand, we need to check the length of a shortest transformation
that avoids matchings that are not inclusion-wise maximal. For this purpose,
we use a reduction to the problem \DST, which is defined as follows.

\begin{quote}
	\textsc{\DST}\\
	\textbf{Input:} Directed graph $D = (V, A)$, integral arc weights $c \in \mathbb{Z}_{\geq 0}^A$, root vertex $r \in V$, and terminals ~$T \subseteq V$.\\ 
	\textbf{Task:} Find a minimum-cost directed tree in $D$ that connects the root $r$ to each terminal.
\end{quote}

It is known that \DST parameterized by the number of terminals is
FPT~\cite{Bjorklund:07,DW:71}. For the remainder of this section let $d :=
|M_s \symdiff M_t|$. Our reduction gives at most $d/2$ terminals. As a
consequence, we may use the FPT algorithm from~\cite{Bjorklund:07} for \DST to
obtain the following result.

\begin{lemma}
    Let $M_s$ and $M_t$ be maximum. Then there is an algorithm that finds in
    time $2^{d/2} \cdot n^{O(1)}$ a shortest transformation from $M_s$ to
    $M_t$, or indicates that no such transformation exists.
    \label{lemma:fpt:max}
\end{lemma}

In order to deal with matchings that are not maximum, we do the following. Let
$(U, W)$ be a bipartition of the vertex set $V$. For each cycle $C$ in the in
the symmetric difference of $M_s$ and $M_t$, we have to guess if $C$ is
reconfigured using an alternating path from an exposed vertex in $U$ or in $W$ in a shortest transformation.
For paths in the symmetric difference there is no choice.  Therefore, the number
of different choices for reconfiguring all paths and cycles in the symmetric difference is bounded by $2^{d/4}$, so we may check all possibilities and pick a
shortest transformation.  Theorem~\ref{thm:fpt} then follows from
lemmas~\ref{lemma:fpt:non-iwm} and~\ref{lemma:fpt:max}.  If $M_s$ is maximum,
then Lemma~\ref{lemma:fpt:max} only needs to be invoked once. By using the
approximation algorithm for \DST from~\cite{Charikar:99} instead of the exact
algorithm from~\cite{Bjorklund:07} in the proof of Lemma~\ref{lemma:fpt:max},
we obtain Corollary~\ref{cor:approximation}.  

Our techniques are not likely to generalize to matchings in non-bipartite
graphs.  We leave as an open problem whether finding a shortest transformation
between two matchings in non-bipartite graphs is FPT in the size of the
symmetric difference of source and target matchings.

In the following, let $\mathcal{C}$ (resp., $\mathcal{P}$) be the set of $(M_s,
M_t)$-alternating cycles (resp., $(M_s, M_t)$-alternating paths) in $(V, M_s
\symdiff M_t)$.

\subsection{Reduction to \DST}
\label{sec:fpt:construction}

Let $I = (G, M_s, M_t)$ be an instance of \MDIST, where the $G = (U \cup W, E)$
is bipartite and the matchings $M_s$ and $M_t$ are maximum.  We will reduce the
task of finding a shortest transformation from $M_s$ to $M_t$ to the \DST
problem. Note that if some edge is not contained in a maximum matching of $G$,
we cannot use it for reconfiguration. Therefore, we may assume that every edge
of $G$ is contained in some maximum matching. Let $X_s$ be the set of
$M_s$-free vertices of $G$. Since $M_s$ and $M_t$ are maximum, we may assume
the following.

\begin{fact}
	We may assume that $X_s \subseteq U$.
	\label{fact:fpt:oneside}
\end{fact}

The main feature of the reduction 
is that the arc-costs in the auxiliary directed graph will reflect the number 
of exchanges we need to perform in the
corresponding transformation. The terminals are placed each cycle in $M_s
\symdiff M_t$ to ensure a correspondence between Steiner trees and
transformations from $M_s$ to $M_t$.  Without loss of generality, let $X_s
\subseteq U$ be the set of $M_s$-free vertices of $G$.   We construct an instance $(D, c, r, T)$ of \DST as follows.  The digraph
$D = (U', A)$ is defined by 
\[
\begin{aligned}
    U'   &:= \{ v \in U \mid \exists \text{ an even-length $M_s$-alternating path from $X_s$ to $v$} \} \cup \{r\}\\
    A  &:= \{ uw \mid u, w \in U,\, \exists v \in W :\, uv \in E \setminus M_s,\,vw \in M_s \} \cup R \enspace,
\end{aligned}
\]
where $r$ is a new vertex and $R := \{ rv \mid v \in X_s\}$. For an arc $uw
\in A$, let the weight $c_{uw}$ be given by
\[
    c_{uw} :=
    \begin{cases}
        0 & \text{if $u = r$},\\
        1 & \text{if there are two edges $uv \in M_t$ and $vw$ in $M_s$.}\\
        2 &  \text{otherwise.}
    \end{cases}
\]
The set $T$ of terminals is given by $T := U' \cap \bigcup_{Z \in \mathcal{C}
	\cup \mathcal{P}} V(Z)$. Note that since $M_s$ and $M_t$ are matchings of $G$,
any two distinct items in $\mathcal{P} \cup \mathcal{C}$ are vertex-disjoint.
The root of the Steiner tree is the vertex $r$. This completes the construction
of the instance $(D, c, r, T)$.

\subsection{Proof of Lemma~\ref{lemma:fpt:max}}
\label{sec:fpt:max}

Let $(G, M_s, M_t)$ is an instance of \MDIST, where the $G = (U \cup W, E)$ is a
bipartite graph and the matchings $M_s$ and $M_t$ are maximum.  By
Fact~\ref{fact:fpt:oneside}, we may assume that each $M_s$-free vertex is in $U$.
Furthermore, let $I' = (D, c, r, T)$ be an instance of \DST, constructed
according to Section~\ref{sec:fpt:construction}, where $D = (U', A)$. 

Let us first introduce some notation that will be used throughout this section.
An arc of weight one is called \emph{special}. 
Note that, by definition of $A$, all the arcs entering in $w$ have to pass
through the vertex $v$ which is matched to $w$ by $M_s$.
In particular, for each pair of edges $uv \in M_t$ and $vw \in M_s$, there exists a
(unique) special arc $uw$ of $D$. 
Finally note that, since a vertex $u$
is incident to at most one edge of $M_s$ and one edge of $M_t$, the vertex $u$
has at most one incoming special arc and at most one outgoing special arc.
For each special arc
$uw \in A$, there is some $Z \in \mathcal{C} \cup \mathcal{P}$, such that $u, w
\in V(Z)$.  Let $Z \in \mathcal{P} \cup \mathcal{C}$. We denote by $A(Z) \in A$
the set of special arcs with both endpoints in $V(Z)$.
If $Z$ is a cycle of length $2k$ then $A(Z)$
is a directed cycle of length $k$ and if $Z$ is a path of length $2k$ then
$A(Z)$ is a directed path of length $k$.  

An arc of weight zero is called \emph{artificial}.  By definition, any
artificial arc is incident to the root $r$.
We first observe some properties of optimal Steiner trees for $I'$.

\begin{proposition}[$\ast$]
	Any optimal Steiner tree $F$ for $I'$ satisfies:
	\begin{enumerate}[label=(\roman*)]
		\item For each $P \in \mathcal{P}$, the tree $F$ contains all arcs in $A(P)$. \label{itm:dst:patharcs}
		\item For each $C \in \mathcal{C}$, the tree $F$ misses exactly one arc of $A(C)$. \label{itm:dst:cyclearcs}
		\item For each $P \in \mathcal{P}$, the root $r$ is joined to the
		$M_s$-free vertex of $P$. \label{itm:dst:pathroot}
	\end{enumerate}
	\label{prop:dst:assumptions}
\end{proposition}

The next lemma shows how to construct from an optimal Steiner tree $F$ of cost
$c(F)$ a reconfiguration sequence of length $c(F)$.

\begin{lemma}
        Let $F$ be an optimal Steiner tree for $I'$. Then there is a
        reconfiguration sequence of length $c(F)$ that transforms $M_s$ to
        $M_t$.
	\label{lemma:dst:sequence}
\end{lemma}
\begin{proof}
	Let $A'$ be the arc-set of $F$.
	Let us consider a depth-first-search (DFS) traversal of the tree $F$ starting
	at $r$, where at each vertex $v$, we select a successor that is joined to~$v$
	by an arc of largest weight among all successors of $v$ in $F$ that have not
	been visited previously by the DFS. The traversal of $F$ yields a sequence 
	$a_1, a_2, \ldots, a_m$ of arcs of $A'$ of length $m := 2|A'|$. This
	is the case since each arc is visited twice: once when going ``down'' in the tree
	$F$ and a second time when going ``up'', that is, when we are backtracking.
	Note that the preference for arcs of weight two will be important to prove
	the correctness of the transformation.
	
        For $1 \leq i \leq m$, let $u_i, w_i \in U$, such that $a_i = u_iw_i$.
        If $a_i$ is not artificial, let $P_i = u_i, v_i, w_i$ be the unique
        path of length two in $G$, such that $v_iw_i \in M_s$. Furthermore, if
        $a_i$ is not artificial, let $\bar e_i :=
        u_iv_i$ and $e_i := v_iw_i$.  Note that for $1 \leq i \leq m$, if $a_i$
        not artificial, then $e_i \in M_s$ and $\bar e_i \in  E
        \setminus M_s$.  By Proposition~\ref{prop:dst:assumptions}, for each $C
        \in \mathcal{C}$, the tree $F$ misses exactly one arc $a_C$ of $A(C)$.
        By definition, the arc $a_C$ is a shortcut for a path $u, v, w$ of
        length two in $G$, such that $uv \in M_t$ and $vw \in M_s$.  We denote by
        $e_C$ the unique $M_t$-edge incident to the source of $a_C$.  We will use
        the edge $e_C$ when backtracking to complete the reconfiguration of the
        cycle~$C$.  
	
	We now specify the reconfiguration sequence $M_0, M_1, \ldots, M_m$, where
	$M_0 := M_s$. For each $1 \leq i \leq m$, let $M_{i+1}$ be given as follows. 
	
	\begin{equation}
	M_{i+1} :=
	\begin{cases}
	M_{i} 			& \text{if $a_{i+1}$ is artificial}\\
	M_{i} - e_{i+1} + \bar e_{i+1}& \text{otherwise, if we traverse $a_{i+1}$ downwards},\\
	M_{i} 			& \text{otherwise, if $a_{i+1}$ is special}, \\
	M_{i} - \bar e_{i+1} + e_C	& \text{otherwise, if $a_{i} \in A(C)$ for $C \in \mathcal{C}$},\\
	M_{i} - \bar e_{i+1} + e_{i+1}	& \text{otherwise}.
	\end{cases}
	\label{eq:Mi}
	\end{equation} 
	
	Note that if several cases apply, then we choose the first one in the given
	order.  
	The next two claims establish that $M_0, M_2, \ldots, M_m$ is a
	reconfiguration sequence from $M_s$ to $M_t$.  Claim~\ref{claim:indep}
	is omitted due to space limitations.

    \setcounter{claim}{0}
    \begin{claim}\label{claim:indep}
	    Let $1 \leq j \leq m$ and suppose we have processed arcs $a_1, a_2,
	    \ldots, a_{j}$, so $M_j$ is our current set of edges. Let $a_i \in A$,
	    such that $a_i$ is not artificial. Then the following holds: 
	    \begin{itemize}
		    \item If $a_i$ has been traversed downwards but not upwards, then $\bar e_i=u_iv_i$ is in $M_i \setminus M_i$.
		    \item If $a_i$ has not been traversed so far, then $e_i=v_iw_i$ is in $M_i \cap M_i$.
		    \item If an artificial arc $rv$ has not been visited so far, then $v$ is $M_j$-free.
		    \item Let $C$ be a cycle of $\mathcal{C}$ and let $a_C=uw$ be the arc of $C$ not in $F$ and $a=u'u$ the arc before $a_C$ in $C$ and $a'=ww'$ the arc after it in $C$. Then if at step $j$, $a$ has been visited upwards and $a'$ has not then $M_j$ is $u$-free.
	    \end{itemize}
    \end{claim}

	\begin{claim}
		$M_m = M_t$.
		\label{claim:dst:matchingtarget}
	\end{claim}

	From the definition of~\eqref{eq:Mi} it is straightforward to see that
	all the artificial arcs lead to no exchange, all the special arcs lead to
	one exchange and all the other arcs lead to two exchanges. Therefore, the
	total number of exchanges is $c(F)$, as claimed.
\end{proof}

We show how to obtain from a reconfiguration sequence $\mathcal{S}$ that
transforms $M_s$ to $M_t$ a subgraph of $D_\mathcal{S}$ of $D$, such that the
root $r$ is connected to each terminal in $D$. Let $\mathcal{S} := M_0, M_1,
\ldots, M_m$ be a reconfiguration sequence of matchings
of $G$. Let $A(\mathcal{S})$ be the subsets of the arcs of $D$ that correspond
to some exchange in $\mathcal{S}$. That is, $A(\mathcal{S})$ is given by
\[
    A(\mathcal{S}) := \{ uw \in A \mid \exists 1 \leq i \leq m, v \in U':\, \{uv, vw\} = M_{i-1} \symdiff M_i \}
\]
Furthermore, let $D_{\mathcal{S}} := (U', A(\mathcal{S}) \cup R) \subseteq D$
be the subgraph of $D$ given by the arcs $A(\mathcal{S})$ and the artificial
arcs $R$.

\begin{lemma}[$\ast$]
        Let $\mathcal{S} := M_0, M_1, \ldots, M_m$ be a reconfiguration
        sequence from $M_s$ to $M_t$. Then, for each terminal $w \in T$, there
        is a $rw$-path in $D_{\mathcal{S}}$.
	\label{lemma:dst:connectivity}
\end{lemma}

The main idea in the proof of Theorem~\ref{lemma:fpt:max} is as follows.  Given some
reconfiguration sequence that transforms $M_s$ into $M_t$, we obtain a subgraph
of $D$ that contains a Steiner tree $F$.  Let $F^*$ be an optimal solution to
$I'$. By Lemma~\ref{lemma:dst:sequence}, we obtain from $F^*$ a reconfiguration
sequence of length $m^* = c(F^*)$ that transforms $M_s$ to $M_t$. Hence, we
have $m^* = c(F^*) < c(F) \leq m$.  Therefore, the directed Steiner tree $F$ is
optimal if and only if $\mathcal{S}$ is a shortest reconfiguration sequence.
Using the exact \DST algorithm from~\cite{Bjorklund:07} to compute $F^*$
yields Theorem~\ref{thm:fpt}, while using the approximation algorithm from~\cite{Charikar:99}
yields Corollary~\ref{cor:approximation}.

\subsection{Proof of Theorem~\ref{thm:fpt}}
\label{sec:fpt:main}

By Lemma~\ref{lemma:fpt:non-iwm}, we may compute a shortest transformation from
$M_t$ to $M_t$ via a matching that is not inclusion-wise maximal in polynomial
time.
Hence it remains to deal with the case that such a transformation is expensive
(or does not exist). 
For this purpose, we reduce the problem to the case where both matchings are
maximum and repeatedly use the algorithm from Lemma~\ref{lemma:fpt:max}.  Let
us fix some bipartition $(U, W)$ of the vertex set $V$.  If $M_s$ is not
maximum, then there is an $M_s$-augmenting path, so there are $M_s$-free
vertices on both sides of the bipartition $(U, W)$. In particular, it may
happen that $M_s$-free vertices on both sides can be used in order to
reconfigure a cycle in $M_s \symdiff M_t$, as explained below.  Hence, for each
cycle in $|M_s \symdiff M_t|$, we have to check if in a shortest transformation
it is reconfigured using an exposed vertex from $U$ or from $W$. Since the number of
choices is bounded by a function of the size of the symmetric difference, we just enumerate
all possibilities. 

The reduction to the restriction of \MDIST where both matchings are maximum is as follows.
We recall the following lemma from~\cite{Ito:11}.

\begin{lemma}[{\cite[Lemma 1]{Ito:11}}]
   Suppose that $M_s$ and $M_t$ are maximum matchings of $G$. Then, there
   a transformation from $M_s$ to $M_t$ if and only if, for each
   cycle $C \in \mathcal{C}$, there is an $M_s$-alternating path in
   $G$ connecting an $M_s$-free vertex to $C$.
   \label{lemma:cyclereconf}
\end{lemma}

\noindent In the light of this lemma we say that a cycle $C \in
\mathcal{C}$ is \emph{reconfigurable from $U$} if there is an
$M_s$-alternating path that connects an $M_s$-free vertex in $U$ to $C$. 
If a cycle $C \in \mathcal{C}$ is neither reconfigurable from $U$ nor from $W$,
then there is no transformation from $M_s$ to~$M_t$.
Furthermore, we say that a path $P \in \mathcal{P}$ is reconfigurable from $U$,
if $P$ contains an $M_s$-free vertex in $U$.
Let us fix for each $F \in \mathcal{C} \cup \mathcal{P}$ a choice $S_F \in \{U,
W\}$, such that $F$ is reconfigurable from $S_F$. 
Let $S \in \{U, W\}^{\mathcal{C} \cup \mathcal{P}}$ be a tuple corresponding
to feasible choice of the side of the bipartition for each item in $\mathcal{C}
\cup \mathcal{P}$. Furthermore, let $X_s$ be the set of $M_s$-free vertices of
$G$.  Based on the choice $S$, we construct two instances $I_U(S)$ and $I_W(S)$ of
\MDIST on maximum matchings as follows.  Let
\[
  M_U(S) := \bigcup_{F \in \mathcal{C}\cup\mathcal{P} :\, S_F = W} (E(F) \cap M_s)  \quad\cup \bigcup_{F \in \mathcal{C}\cup\mathcal{P} :\, S_F = U} (E(F) \cap M_t)  \enspace,
\]
and let $I_U(S) := (G - (X_s \cap W), M_s, M_U(S))$ and $I_W(S) := (G - (X_s \cap U),
M_U(S), M_t)$. Note that $M_s$ and $M_U(S)$ are maximum in $G - (X_s \cap W)$
and $M_U(S)$ and $M_t$ are maximum in $G - (X_s \cap U)$.

Let us now describe the algorithm that outputs a shortest transformation from
$M_s$ to~$M_t$ if it exists.
For each feasible choice $S \subseteq \{U, W\}^{\mathcal{C} \cup \mathcal{P}}$
we invoke Lemma~\ref{lemma:fpt:max} on $I_U$ (resp., $I_W$)  to obtain a
transformation $\alpha_1$ (resp., $\alpha_2$) from $M_s$ to $M_U(S)$ (resp.,
$M_U(S)$ to $M_t$) if it exists. Clearly, by combining $\alpha_1$ and
$\alpha_2$ we obtain a transformation from $M_s$ to $M_t$. Let us denote this
transformation by $\alpha(S)$ and by $|\alpha(S)|$ its length.
Let $S^* \in \{U, W\}^{\mathcal{P}\cup\mathcal{C}}$ such that $|\alpha(S^*)|$
is minimal. If there is no transformation for any choice $S$, then
$|\alpha(S^*)| = \infty$.
According to Lemma~\ref{lemma:fpt:non-iwm} we determine in polynomial
time the length $|\beta^*|$ of a shortest transformation from $M_s$ to $M_t$
via a matching that is not inclusion-wise maximal. 
If $|\alpha(S^*)|$ and $|\beta^*|$ are both not finite, then $(G, M_s, M_t)$
must be a \no-instance. 
Otherwise, we output output either the transformation of length $|\beta^*|$ via
a matching that is not inclusion-wise maximal, or a transformation of length
$|\alpha(S^*)|$ via $M_U(S^*)$, depending on which is shorter.

The running time of the algorithm is dominated by the computation of $|\alpha(S^*)|$
and $S^*$. Both $|\alpha(S^*)|$ and $S^*$, as well as a corresponding
transformation if it exists can be computed in time $2^{d/2} \cdot
2^{d/2}\cdot n^{O(1)}$ by Lemma~\ref{lemma:fpt:max}. The leading factor of
$2^{d/2}$ bounds the number of feasible choices of $S \in \{U, W\}^{\mathcal{P}
\cup \mathcal{C}}$. The overall running time of the algorithm is therefore
bounded by $2^d \cdot n^{O(1)}$ as claimed.
The correctness of the algorithm is a consequence of the following lemma, which
ensures that in the case that no optimal transformation visits a matching that
is not inclusion-wise maximal, we have that an optimal choice of $S\subseteq
\{U, W\}^{\mathcal{C} \cup \mathcal{P}}$ also results in an optimal
transformation from $M_s$ to $M_t$.

\begin{lemma}[$\ast$]
  Suppose that no shortest transformation from $M_s$ to $M_t$ has an 
  intermediate matching that is not inclusion-wise maximal.
  Let $\tau := M_0, M_1, \ldots, M_m$ be a shortest transformation from $M_s$
  to $M_t$ of length $m$. Then $|\alpha(S^*)| \leq m$.
  \label{lemma:fpt:correctness}
\end{lemma}

%% file: 3_exact_distance.tex
\section{Exact Distance and Diameter}
\label{sec:diameter}

We will consider the exact versions of the problems of
determining the distance of two matchings and the diameter of the
reconfiguration graph. 
Before presenting our hardness result for these problems, let us first prove
that we can determine in polynomial if the diameter of the reconfiguration
graph of matchings is finite.  In other words, for any $k \geq 0$, we can
determine in polynomial time if ${\matchgraph(G)}$ is connected.

Note that proof of~\cite[Proposition 2]{Ito:11} provides a polynomial-time
algorithm for deciding, whether the distance of two given matchings is finite.
We show that there is a polynomial-time algorithm that decides if the diameter
of $\matchgraph(G)$ is finite. To this end we need a condition that
characterizes \yes-instances of \MR, which was given in~\cite[Lemma 1]{Ito:11},
as well as the \emph{Edmonds-Gallai decomposition} of the vertex set of
$G$~\cite[Ch.~24.4b]{Schrijver:CO}.  
Using these two ingredients, we can check efficiently if the diameter of
$\matchgraph(G)$ is finite.

\begin{theorem}[$\ast$]
  There is a polynomial-time algorithm that, given a graph $G$ and a number $k
  \in \nat$, decides if $\matchgraph(G)$ is connected.
  \label{thm:connected}
\end{theorem}

For the remainder of this section we will consider maximum matchings. That is, we
restrict our attention to instances of \EMDIST and \EMDIAM, where the number $k$
is the equal to the size of a maximum matching of the input graph $G$. Using a similar
construction to the one from Section~\ref{sec:inapprox:construction}, we show
that the problem of \EMDIST is complete for the class $\D^\P$, which was
introduced in~\cite{PY:84} as the following class of languages
\[
  \D^\P := \{ L_1 \cap L_2 \mid L_1 \in \NP,\,L_2 \in \coNP \}
\]
From the $\D^\P$-completeness of \EMDIST, the $\D^\P$-hardness of \EMDIAM
will follow in a relatively straightforward manner.  To show that \EMDIST is
$\D^\P$-hard we use a reduction from \textsc{Exact Vertex Cover}, which is
given a follows.

\begin{quote}
  \textsc{Exact Vertex Cover} \\
  \textbf{Input:} Graph $G$ and number $\ell$.\\
  \textbf{Question:} Is $\ell$ the minimum size of a vertex cover? 
\end{quote}

\paragraph*{The Construction.} 
Let $H = (V, E)$ be a graph. Starting with the empty graph, we construct a graph $G$ and two matchings $M$ and $N$ of $G$ as follows. 
We first create a vertex $t$. For each $e \in E$,
we create a 4-cycle $C_e$ on the vertices $c_e^1, c_e^2, c_e^3, c_e^4$.  Then, for
each $v \in V$, we create two vertices~$p_v^1$ and $p_v^2$ and create a path $P_v$ $p_v^1 p_v^2 t$.  
Now, for $e \in E$ and $v \in V$, we add an edge $c_e^1p_v^1$ whenever $e$ is
incident to $v$. In a final step we create new vertices $q_1,\ldots,q_6$ and
add a path $t,q_1,q_2,\ldots,q_6$ to $G$. Note that we added this path in oder
to avoid a case analysis when we determine the maximal distance of any two
matchings. Observe that $G$ is bipartite and that a maximum matching
in $G$ has precisely one unmatched vertex. Let $M$ and $N$ be two maximum
matchings that both leave $t$ exposed and are disjoint on $C_e$ for each $e \in
E$. This completes the construction of $G$, $M$, and $N$.  An illustration of
the construction where $H$ is a cycle on three vertices is shown in
Figure~\ref{fig:diamexample}. The dashed edges are in $M$ and the dotted edges
are in $N$. 

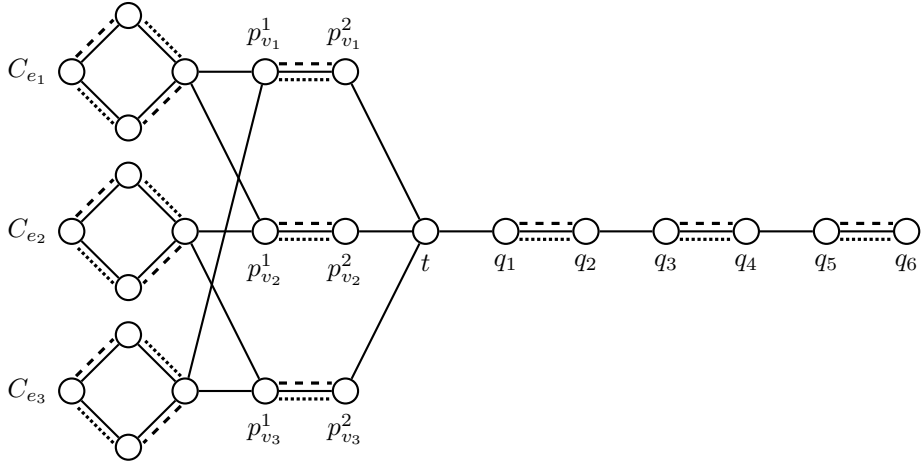
\begin{figure}[t]
  \centering
  \begin{tikzpicture}[vertex/.style={shape=circle,thick,draw,node
    distance=3em},terminal/.style={shape=rectangle,thick,draw,node
    distance=3em},medge/.style={very thick,dashed},nedge/.style={very
    thick,densely dotted}]
    \node[vertex]			(c11)	{};
    \node[vertex,below left of=c11]	(c12)	{};
    \node[vertex,above left of=c12,label=left:$C_{e_1}$]	(c13)	{};
    \node[vertex,above right of=c13]	(c14)	{};

    \node[vertex,below=5em of c11]	(c21)	{};
    \node[vertex,below left of=c21]	(c22)	{};
    \node[vertex,above left of=c22,label=left:$C_{e_2}$]	(c23)	{};
    \node[vertex,above right of=c23]	(c24)	{};

    \node[vertex,below=5em of c21]	(c31)	{};
    \node[vertex,below left of=c31]	(c32)	{};
    \node[vertex,above left of=c32,label=left:$C_{e_3}$]	(c33)	{};
    \node[vertex,above right of=c33]	(c34)	{};

    \node[vertex,right of=c11,label=above:{$p_{v_1}^1$}] (p11) {};
    \node[vertex,right of=c21,label=below:{$p_{v_2}^1$}] (p21) {};
    \node[vertex,right of=c31,label=below:{$p_{v_3}^1$}] (p31) {};
    \node[vertex,right of=p11,label=above:{$p_{v_1}^2$}] (p12) {};
    \node[vertex,right of=p21,label=below:{$p_{v_2}^2$}] (p22) {};
    \node[vertex,right of=p31,label=below:{$p_{v_3}^2$}] (p32) {};
    \node[vertex,right of=p22,label=below:{$t$}] (t) {};
    \node[vertex,right of=t,label=below:{$q_1$}] (q1) {};

    \foreach \s in {2, ..., 6} {
      \pgfmathparse{int(\s-1)}
      \edef\prv{\pgfmathresult}
      \node[vertex,right of=q\prv,label=below:{$q_\s$}] (q\s) {};
      \draw[thick] (q\prv) -- (q\s);
    }

    \foreach \s in {1,...,4} {
      \pgfmathparse{int(Mod(\s,4)+1)}
      \edef\nxt{\pgfmathresult}
      \draw[thick] (c1\s) -- (c1\nxt);
      \draw[thick] (c2\s) -- (c2\nxt);
      \draw[thick] (c3\s) -- (c3\nxt);
    }

    \draw[thick] (c11) -- (p11) -- (c31);
    \draw[thick] (c11) -- (p21) -- (c21);
    \draw[thick] (c21) -- (p31) -- (c31);
    \draw[thick] (p11) -- (p12) -- (t);
    \draw[thick] (p21) -- (p22) -- (t);
    \draw[thick] (p31) -- (p32) -- (t);
    \draw[thick] (t) -- (q1);

    \foreach \s in {1, 2, 3} {
        \draw[medge] ([xshift=2pt,yshift=-2pt]c\s1.south west) -- ([xshift=2pt,yshift=-2pt]c\s2.north east);
        \draw[medge] ([xshift=-2pt,yshift=2pt]c\s3.north east) -- ([xshift=-2pt,yshift=2pt]c\s4.south west);
        \draw[nedge] ([xshift=2pt,yshift=2pt]c\s1.north west)  -- ([xshift=2pt,yshift=2pt]c\s4.south east);
        \draw[nedge] ([xshift=-2pt,yshift=-2pt]c\s2.north west) -- ([xshift=-2pt,yshift=-2pt]c\s3.south east);
    }
    \foreach \s in {1,2,3} {
        \draw[medge] ([yshift=3pt]p\s1.east) -- ([yshift=3pt]p\s2.west);
        \draw[nedge] ([yshift=-3pt]p\s1.east) -- ([yshift=-3pt]p\s2.west);
    }
    \foreach \s in {1,3,5} {
        \pgfmathparse{int(\s+1)}
        \edef\nxt{\pgfmathresult}
        \draw[medge] ([yshift=3pt]q\s.east) -- ([yshift=3pt]q\nxt.west);
        \draw[nedge] ([yshift=-3pt]q\s.east) -- ([yshift=-3pt]q\nxt.west);
    }
  \end{tikzpicture}
  \caption{Example of the construction of $G$, $M$, $N$ from the graph $H = C_3$.}
  \label{fig:diamexample}
\end{figure}

\paragraph*{Proof of Theorem~\ref{thm:dpc} (sketch)}
By~\cite[Theorem 5.4]{Cai:88}, the problem \textsc{Exact Vertex Cover} is complete
for the class $\D^\P$. Theorem~\ref{thm:dpc} now follows directly from this result
and the next two lemmas. The key insight is that if we know the size of a
smallest vertex cover, then we know the distance of the two matchings
and the diameter of the reconfiguration graph after performing the construction
above.

\begin{lemma}[$\ast$]
  \label{lemma:dist:cover}
  Let $H$ be a graph and let $G$, $M$, $N$ be the graph and the two matchings
  obtained according to the construction above. Then $\dist_{\matchgraph(G)}(M,
  N) = 3|E(H)| + 2\tau(H)$, where $\tau(H)$ is the size of a smallest vertex
  cover of $H$.
\end{lemma}

\begin{lemma}[$\ast$]
  \label{lemma:diam:ub}
  Let $H$ be a graph and let $G$, $M$, $N$ be the graph and the two matchings
  obtained according to the construction above. Then
    $\diam_{\matchgraph(G)} = \dist_{\matchgraph(G)}(M,N)+6$.
\end{lemma}

%% file: appendix1.tex
\section{Edmonds-Gallai Decomposition}

A matching of a bipartite graph $(A \cup B, E)$ is called \emph{$A$-perfect},
if it matches each vertex of $A$.  
Let $G = (V, E)$ be a graph, let $M$ be a maximum matching of $G$ and let $X :=
\{ v \in V(G) \mid v \text{ is $M$-free}\}$.
Consider the following partition of the vertex set $V(G)$ into $D(G), A(G),
C(G) \subseteq V(G)$. 

\begin{itemize}
    \item $D(G) := \{ v \in V(G) \mid \text{there is an $M$-alternating path of even length from $X$ to $v$} \}$
    \item $A(G) := \{ v \in V(G) \setminus D(G) \mid \text{there is an $M$-alternating path from $X$ to $v$} \}$ 
    \item $C(G) := \{ v \in V(G) \mid \text{there is no $M$-alternating path from $X$ to $v$} \}$
\end{itemize}

The following classical theorem states that $D(G)$, $A(G)$, and $C(G)$ depend
only on $G$, but not on the choice of the maximum matching.

\begin{theorem}[{Edmonds-Gallai decomposition, see~\cite[Ch.~24.4b]{Schrijver:CO}}]
    Let $G$ be a graph and $D(G)$, $A(G)$, $C(G)$ be given as above.  Then, a maximum
    matching of $G$ can be partitioned into
    \begin{enumerate}
        \item a perfect matching of $G[C(G)]$,
        \item a matching that leaves precisely one vertex unmatched in each component of $G[D(G)]$,
        \item and an $A$-perfect matching from $A(G)$ to $D(G)$.
    \end{enumerate}
    Furthermore, the partition of $G$ into $D(G),A(G),C(G)$ can be found in polynomial time.
    \label{thm:egd}
\end{theorem}

\section{Proofs Omitted from Section~\ref{sec:inapprox}}

\begin{proof}[Proof of Lemma~\ref{lemma:ptime:msr}]
  Let $I = (G, M_1, M_2)$ be an instance of \MDIST where at least one of $M_1$
  and $M_2$ is not inclusion-wise maximal. Without loss of
  generality, we assume that $M_1$ is not inclusion-wise maximal. Since $M_1$
  and $M_2$ are matchings, the graph on the edges $M_1 \symdiff M_2$ is composed of
  disjoint paths and cycles. Let $\opt(I)$ be the shortest length of a
  reconfiguration sequence that transforms $M_1$ into $M_2$. We show that
  $\opt(I)$ is either $|M_1 \symdiff M_2|/2$ or $|M_1 \symdiff M_2|/2 + 1$.
  The exact value can be determined as follows. The shortest reconfiguration
  sequence is equal to $|M_1 \symdiff M_2|/2$ if and only if one of the
  following statements is true.
  \begin{itemize}
   \item All the components of $M_1 \symdiff M_2$ are paths;
   \item One of the paths in $M_1 \symdiff M_2$ has an odd number of edges.
  \end{itemize}
  Note that $|M_1 \symdiff M_2|/2$ is indeed a lower bound on the length of the
  reconfiguration sequence. Assume that at least one of the above statements is
  true. If the symmetric difference of $M_1$ and $M_2$ only contains paths,
  then we can move tokens from $M_1$ to $M_2$ in a greedy fashion until the
  target configuration $M_2$ is reached. 
  Assume now that $M_1 \symdiff M_2$ contains a path with an odd number of
  vertices. Let us prove that there is a reconfiguration sequence that
  transforms $M_1$ into $M_2$ by induction on the number of cycles in $M_1
  \symdiff M_2$.  Since $|M_1| = |M_2|$, there is a path in $M_1 \symdiff M_2$
  containing more edges of $M_2$ than edges of $M_1$. We slide tokens on this path
  path as in the previous case, until a path $P$ consisting of a single
  $M_2$-edge remains in the symmetric difference. 
  Let $C$ be a cycle in $M_1 \symdiff M_2$. We can transform $C$ into a path,
  by moving a token from $C$ to $P$. This decreases the number of cycles in the
  symmetric difference by one and leaves a path that contains more $M_2$-edges
  than $M_1$-edges. We invoke the induction hypothesis and the statement
  follows.
 
  Assume now that $M_1 \symdiff M_2$ contains a cycle and does not
  contain a path with an odd number of edges. If there is a transformation
  of length $|M_1 \symdiff M_2|/2$, then, at each step, we must move a token
  from $M_1$ to $M_2$. We can do so, again in a greedy fashion, by moving
  tokens on paths in $M_1 \symdiff M_2$. At some point however, only cycles
  remain the symmetric difference. But then, we cannot move any token, such
  that the number of tokens on $M_2$ is increased. Therefore, there is no
  transformation from $M_1$ to $M_2$ of length $|M_1 \symdiff M_2|/2$.
  
  Let us finally prove that a transformation of length $|M_1 \symdiff M_2|/2+1$
  always exists. Since $M_1$ is not inclusion-wise maximal, it is possible to
  move any token of $M_1$ to some edge $e$, which is not incident to any edge
  of $M_1$. Let $f$ be an edge of $M_1$ in a cycle of $M_1 \symdiff M_2$.
  We can move the token of $f$ to $e$. Let $M_1'$ the resulting matching. Note
  that the cycle on which $f$ appeared is now a path of odd length in the
  symmetric difference. Moreover, we have that $|M_1' \symdiff M_2| \leq |M_1
  \symdiff M_2|$. By our previous arguments,
  there is a transformation of length at most $|M_1 \symdiff M_2|/2$ between
  $M_1'$ and $M_2$. It follows that there is a reconfiguration sequence of
  length $|M_1 \symdiff M_2|/2+1$ that transforms $M_1$ into $M_2$. 
  
  All the steps can indeed by performed in polynomial time, which concludes the
  proof.
\end{proof}

\begin{proof}[Proof of Lemma~\ref{lemma:inapprox:ub}]
	As usual for the token sliding operation, we will frequently say that we
	\emph{reconfigure the alternating path $P$} by ``sliding'' the matching edges
	one-by-one, until the other end-vertex of the path $P$ is $M_1$-free. In this
	procedure, one exchange is performed per matching-edge.
	
	Let $C = \{S_1, S_2, \ldots, S_k\}$ be a cover of $U$ of size $k = \opt(I)$.
	We construct a transformation from $M_1$ to $M_2$ of the desired size. For $1
	\leq i \leq k$, let $T_i \subseteq U$ be the items that are covered by $S_i$
	but not by $S_j$, $1 \leq j < i$. Note that since $C$ is a set cover, $U = \cup_{i=1}^k T_i$.
	
	Let us present a way to reconfigure $C_u \cap M_1$ to $C_u \cap 
	M_2$ for each $u \in T_i$ as follows in such a way, after these operations the resulting matching
	still contains the edges of $M'$ (recall that $M'$ is the matching without the edges of the $C_4$'s).
	First, we reconfigure the $M_1$-alternating path from $q_{S_i}^{1}$ to the $M_1$-free vertex $q_{S_i}^{L}$,
	which takes $L$ steps and leaves vertex $q_{S_i}^{1}$ exposed.
	Now, for each $u \in T_i$, we have the following path $P'$ in $G$: the subpath of $P_u$ from $p_u^1$ until the terminal vertex
	$p_u^j$ that is connected to some vertex $p_{S_i}^{j'}$ and then the subpath of $P_S$ from $p_{S_i}^{j'}$ to $q_{S_i}^1$.
	Note that by definition of $M'$ this path is an alternating path.
	We reconfigure the alternating path $P'$, which leaves $p_u^1$ exposed and takes at
	most $2(|S_i|+f_u)$ steps. Since $p_u^1$ is exposed, we can reconfigure the cycle $C_u \cap M_1$ to $C_u
	\cap M_2$ with three exchanges such that $q_u^1$ is again exposed. We finally
	undo the changes on $P'$ and have again $p_{S_i}^1$ exposed. We repeat this operation for every $u \in T_i$.
	We finally reconfigure back the path $Q_{S_i}$. Note that after all these steps, the resulting matching still
	contains the edges of $M'$.
	
	Let us now count the number of steps we performed to reconfigure $C_u \cap M_1$ to $C_u \cap 
	M_2$ for each $u \in T_i$. The reconfiguration of $Q_{S_i}$ costs  $2L$ steps 
	($L$ steps at the beginning and $L$ steps at the end to put it back). 
	Moreover, for each $u \in T_i$ we perform at most $3 + 2(|S|+f_u)$ exchanges. 
	
	So if we repeat this operation for every $S_i$, we know that we will reconfigure
	the matching on every $C_4$ since $\bigcup S_i = U$. So the total number of
	exchanges we performed is at most
	\[
	2L|C| + \sum_{1 \leq i \leq k} |T_i| (3 + 2(f+|S_i|)) \leq  2L|C| + 2|U| (2+f+d))
	\]
	as claimed.
\end{proof}

\begin{proof}[Proof of Lemma~\ref{lemma:inapprox:lb}]
	Since $M_1$ and $M_2$ are maximum, any $u \in U$ a transformation from $M_1$ to $M_2$ must
	have $p_u^1$ exposed at some point. Indeed if a matching edge on the
	cycle $C_u$ is moved, it has to be moved an edge incident to it and then, by construction
	on an edge incident to $p_u^1$. 
	
	Let us consider a transformation from $M_1$ to $M_2$.  Given a matching $M$,
	we say that a set $S \in \sets$ is \emph{active} if $q_{S}^{L}$ is covered by $M$. 
	Let $C \subseteq \sets$ be the sets that are active with respect to some matching of the
	transformation $\tau$. Let us prove that the sets $C$ cover the items $U$.
	Let $u \in U$. As we already mentioned, for each $u \in U$, the vertex $p_u^1$ has
	to be exposed at some step $M$ of the transformation. Let us prove that a set $S \in
	\sets$ containing $u$ is active with respect to $M$. 
	We can then define an
	$M$-alternating path starting on $p_u^1$ that alternates between edges in $E
	\setminus M$ and edges in $M$. Since $P_u$ is of even length and $p_u^1p_u^2$ is not in $M$
	either we have an augmenting path included in $P_u$ (a contradiction to the maximality of $M$) or this alternating path must leave $P_u$. By
	the construction of $G$, such an edge is of the form $e:=p_u^ip_S^j$ where $i$ is even and $j$ is odd. Note
	moreover, that 
	Let us follow the $M$-alternating path along $P_S$ and $Q_S$ until it is no longer
	possible. Three cases may occur: (i) the alternating path ends at $q_S^L$.  In
	this case $S$ is active and the conclusion holds since by construction $u \in S$. (ii) the alternating path
	ends at some other vertex $w$ of $P_S$ or $Q_S$. Then $w$ is exposed and we
	have thus discovered an augmenting path, a contradiction with the maximality of $M$.
	(iii) the alternating path leaves $P_S$
	via some vertex $p_S^r$ (since vertices of $Q_S$ have degree $2$ (their precesessors and successors in $Q_S \cup P_S$). 
	Then there is some $v \in U$ such that $M$ contains an edge $p_v^sp_S^r$, where $s$
	is even. Since it is no longer possible to follow $P_S$ in $p_v^s$, it means that $p_v^s$ is incident to 
	an edge of $M$ (otherwise we would take $p_S^rp_S^{r+1}$ in the alternating path). Since $e:=p_u^ip_S^j$ where $j$ is odd is in $M$ then $r$ must be even. By construction of $G$, $p_S^r$ has degree $2$ 
	and then there is no edge $p_v^sp_S^r$, a contradiction. So only case (i) can occur, and then a set $S$ containing $u$ is active for $M$.
	
	Let $C \subseteq \sets$ be the subset of
	sets that are active in some step of the transformation $\tau$.  Clearly,
	there is a polynomial-time algorithm $A'$ that outputs $C$ given $\tau$. By
	the discussion above, the set $C$ is a set cover. Let us finally prove that
	the size of $C$ is at most $|\tau| / 2L$.
	For each set $S \in C$ we have to reconfigure an alternating path from
	$q_S^1$ to $q_S^L$ twice (once to expose $q_S^1$ and once to expose again
	$q_S^L$), the output $C$ of $A'$ has cardinality at most $|\tau|/2L$.
\end{proof}

\begin{proof}[Proof of Theorem~\ref{thm:inapprox:msr}]
	Suppose we have an $f(n')$-factor approximation algorithm $A$ for \MDIST, where
	$n' = |I'|$. That is, $A$ computes a reconfiguration sequence from $M_1$ to
	$M_2$ of length at most $f(n')\cdot\opt(I')$.
	Then we can use algorithm $A'$ from Lemma~\ref{lemma:inapprox:lb} to compute
	a $2f(n')$-approximate solution of $I$:
	\begin{align*}
	A'(I) & \leq  \frac{A(I')}{2L} \leq \frac{f(n')\cdot \opt(I')}{2L} \\
	& \leq   \frac{f(n')(\opt(I) \cdot 2L + 2|U| (f+d+2))}{2L} \\
	& \leq  f(n') \opt(I) + 1 \\
	& \leq 2f(n') \opt(I)\enspace,
	\end{align*}
	where $A'(I)$ is the size of the cover $A'$ computes. The first inequality
	follows from Lemma~\ref{lemma:inapprox:lb}, the third inequality from
	Lemma~\ref{lemma:inapprox:ub} and the fourth one from the definition of $L$.
	Suppose for a contradiction that there is a
	polynomial-time $o(\log n')$-approximation algorithm for \MDIST, so we
	let $f(n') = o(\log n')$. By the construction of $G$ and $m =
	\poly(n)$, we have $n' = \poly(n)$. Then, by the
	reasoning above, we have a $2f(n^k)$-approximation algorithm for \textsc{Set
		Cover} for some constant $k$. But
	\[
	\limsup_{n\to\infty} \frac{2f(n^k)}{\log n} = \limsup_{n\to\infty} \frac{2f(n)}{\log n^{1/k}} = \limsup_{n\to\infty} \frac{2k f(n)}{\log n} = \limsup_{n'\to\infty} \frac{2k f(n')}{\log n'} = 0
	\]
	Therefore, a sublogarithmic-factor approximation guarantee for \MDIST implies a
	sublogarithmic-factor approximation guarantee for \textsc{Set Cover}. But
	there is no polynomial-time $(1-\varepsilon)\log n$-factor approximation algorithm
	for \textsc{Set Cover} for any $\varepsilon > 0$ unless $\P =
	\NP$~\cite{DS:14}.
\end{proof}

\begin{proof}[Proof of Corollary~\ref{cor:inapprox:nonmaximum}]
	We modify the construction from Section~\ref{sec:inapprox:construction}
	slightly, by adding a path $R$ $r_1, r_2, \ldots, r_{L'}$
	to the graph $G$ and joining $q_S^1$ to $r_1$, where $L' \geq 3|U|+(L+d+f)|\sets|$
	be odd. Furthermore, we add the edges $r_1r_2, r_3r_4, \ldots,
	r_{L'-2}r_{L'-1}$ to the two matchings $M_1$ and $M_2$. Clearly, the
	distance from $M_1$ or $M_2$ to a matching that is not inclusion-wise
	maximal is larger than the distance between $M_1$ and $M_2$. Hence, it is
	easily verified that lemmas~\ref{lemma:inapprox:ub}
	and~\ref{lemma:inapprox:lb} as well as Theorem~\ref{thm:inapprox:msr} hold
	also for the modified construction.
\end{proof}

%% file: appendix2.tex
\section{Proofs Omitted from Section~\ref{sec:FPT}}

\subsection{Proof of Lemma~\ref{lemma:fpt:non-iwm}}
\label{sec:fpt:non-iwm}

We may check in polynomial time if a transformation from $M_s$ to $M_t$
exists~\cite[Proposion 1]{Ito:11}. So let us assume that such a transformation
exists. Then there is a transformation from $M_s$ to $M_t$ via a matching that
is not inclusion-wise maximal if and only if $M_s$ and $M_t$ are not maximum.
Again, we can check this condition in polynomial time and may assume in the
following that $M_s$ and $M_t$ are not maximum.

\setcounter{claim}{0}
Our first claim follows from the first part of the proof of Lemma~\ref{lemma:ptime:msr}.
\begin{claim}
    If $M_s \symdiff M_t$ contains a path of odd length, then there is a
    polynomial-time algorithm that outputs a shortest transformation from
    $M_s$ to $M_t$.
    \label{claim:fpt:oddpath}
\end{claim}
Hence we may assume that $M_s \symdiff M_t$ contains only
even cycles and paths. 

Observe that by using any $M_s$-augmenting path, we can reconfigure $M_s$
into a matching that is not inclusion-wise maximal. Hence, our task is to
find a cheapest augmenting path with respect to a certain cost measure
that reflects the ``progress'' we make by performing exchanges along a
given augmenting path. The progress is essentially the sum of the length
of the remaining transformation and the length of the augmenting path. Let
us fix a bipartition $(U, W)$ of the vertex set of $G$ and let $X$ be the
set of $M_s$-free vertices in $U$. Consider the following auxiliary
digraph $D = (U', A)$ given by
\begin{align*}
    U' &:= \{ v \in U \mid \exists \text{ an even-length $M_s$-alternating path from $X$ to $v$} \}\\
    A  &:= \{ uw \mid u, w \in U, v \in W :\, uv \in E \setminus M_s,\,vw \in M_s \} \enspace.
\end{align*}
Let $Y := \{ v \in U \mid \text{ $v$ has a neighbor in $G$ that is $M_s$-free} \}$.
Directed $XY$-paths in $D$ are in 1-to-1 correspondence with the
$M_s$-augmenting paths in $G$.  For an $XY$-path $P$ in $D$, we denote by $P_G$
the corresponding 
$M_s$-alternating path in $G$.  Let us define arc-costs~$c \in
\mathbb{Z}^A$.  For an arc $uw \in A$ that corresponds to a path $u, v, w$ in
$G$, let 
\[
    c_{uw} := 
    \begin{cases}
        0  &   \text{if $uv \in M_t \setminus M_s$ and $vw \in M_s \setminus M_t$}, \\
        0  &   \text{otherwise, if $w \in V(F)$ for some $F \in \mathcal{P}$, such that $V(P) \cap Y \neq \emptyset$,}\\
        1  &   \text{otherwise, if $uv \in E \setminus (M_s \cup M_t)$ and $vw \in M_s \setminus M_t$,}\\
        2  &   \text{otherwise, if $uv \in E \setminus (M_s \cup M_t)$ and $vw \in M_s \cap M_t$}\enspace.
    \end{cases}
\]
For an $XY$-path $P$ in $D$, let
\[
    \delta_P =
    \begin{cases}
        1   &   \text{if $(M_s \symdiff E(P_G)) \symdiff M_t$ contains a cycle, and}\\
        0   &   \text{otherwise}\enspace. 
    \end{cases}
\]
Lemma~\ref{lemma:fpt:non-iwm} follows from the next two claims.

\begin{claim}
    Let $P$ be an $XY$-path of cost $c(P)$ in $D$. Then there is a
    transformation from $M_s$ to $M_t$ of length at most $c(P) + |M_s \symdiff
    M_t|/2 + \delta_P$. 
    \label{claim:fpt:ub}
\end{claim}
\begin{proof}[Proof of Claim~\ref{claim:fpt:ub}]
    Our goal is to transform $M_s$ into a matching that is not inclusion-wise
    maximal by using exchanges along $P$.  Let $\ell(P)$ be the number of arcs
    of $P$. 
    For an arc $a$ of $D$, let $u_a, v_a, w_a$ be the path in $G$ that
    corresponds to $a$. Starting from the matching $M_s$, for each arc $a$ on
    $P$, in forward direction, we swap $v_aw_a$ for $u_av_a$ to obtain the
    matching $M$. This amounts to $\ell(P)$ exchanges. Observe that $M$ is not
    inclusion-wise maximal. By using the algorithm from
    Lemma~\ref{lemma:ptime:msr}, a transformation from $M_s$ to $M_t$ via $M$
    has length $\ell(P) + |M \symdiff M_t| / 2 + \delta_{P}$.

    In order to obtain a bound in terms of $c(P)$, we first perform some
    preprocessing. Let $x$ (resp., $y$) be the starting vertex (resp., end
    vertex) of $P$. Suppose that $P$ meets an $(M_s, M_t)$-alternating path
    containing a vertex in $Y$. Let $F$ be the first such path and let $z \in
    U$ be the first vertex of $F$ that is visited by $P$. Then we can
    reconfigure $M_s$ into a matching $M_s'$, such that $z$ has a neighbor that
    is $M_s'$-free  and the symmetric difference between the current matching
    and $M_t$ decreases by two in each step. To this end, we perform exchanges
    along $F$ as follows. Let $y'$ be the vertex of $F$ in $Y$. Then
    there is a path $Q$ from $z$ to $y'$ in $D$ that visits only vertices of
    $F$. We process each arc $a$ of $Q$ in reverse order (from $y'$ to $z'$).
    Then $w_a$ is matched to the vertex $v_a$ by $M_s$.
    Furthermore, the vertex $w_a$ must be matched to a vertex $w'$ by $M_t$,
    such that $w'$ is free with respect to the current matching.  Note that if
    $a$ is the first arc of $Q$ we process, then $w_a$ must be covered by
    $M_t$, since otherwise $M_t$ is not maximal. We exchange $v_aw_a$ for
    $w_aw'$, which decreases the size of the symmetric difference of the
    current matching and $M_t$ by two. Let $M_s'$ be the resulting matching. 

    Now we have that  $|M_s' \symdiff M_t| = |M_s \symdiff M_t| -
    |M_s \symdiff M_s'|$. Let $|M_s \symdiff M_s'| =: 2k$. It now remains to consider the $xz$-subpath
    $P'$ of $P$ in order to reach a matching that is not inclusion-wise
    maximal.  In the case that $P$ does not meet an $(M_s, M_t)$-alternating
    path, we let $P' := P$, $M_s' = M_s$, $k = 0$, and $z := y$. Either way, we
    have that $z$ has an  $M_s'$-free neighbor and we have performed $k$
    exchanges to reach $M_s'$. Recall that an arc into an $(M_s,
    M_t)$-alternating cycle has cost one. It follows that if $\delta_P = 0$ and
    $\delta_{P'} = 1$, then $c(P) \geq c(P') + 1$.  Hence, we have that
    \[
        c(P) + |M_s \symdiff M_t|/2 + \delta_{P} \geq c(P') + |M_s' \symdiff M_t|/2 + \delta_{P'} + k\enspace.
    \]

    We show that $\ell(P') + |(M_s' \symdiff E(P'_G))  \symdiff M_t|/2 = c(P')
    + |M_s' \symdiff M_t|/2$.  We start with the matching $M_s'$ and process
    the arcs of $P'$ in forward direction. We exchange for each arc $a$ of $P'$
    the edges $v_aw_a$ and $u_av_a$.  To prove the equality, consider the
    relative changes to the size of the symmetric difference between the
    current matching and $M_t$. If $a$ has cost two, then it corresponds to an
    exchange that increases the size of the symmetric difference by two.  If
    $a$ has cost one, then it corresponds to an exchange that does not alter
    the size of the symmetric difference.  If $a$ has cost zero  we distinguish
    two cases.  First, suppose that $w_a$ is not on some $(M_s,
    M_t)$-alternating path that contains a vertex in $Y$. Then the
    corresponding exchange decreases the size of the symmetric difference by
    two. Otherwise, due to the preprocessing, we know that $a$ is the final arc
    of $P'$ and its target vertex $w_a$ is matched by $M_s'$ to $v_a$.
    Furthermore, since $M_t$ is maximal, we have that $w_a$ has an $M_s'$-free
    neighbor $w'$, such that $ww' \in M_t \setminus M_s$. So we may exchange
    $v_aw_a$ for $w_aw'$ to turn the current matching into a non-inclusion-wise
    maximal one. At the same time, we reduce the size of the symmetric
    difference with $M_t$ by two.  By counting the number of exchanges and
    keeping track of the size of the symmetric difference of the current
    matching and $M_t$, the claimed equation follows.

    Putting things together, there is a transformation from 
    $M_s'$ to $M_t$ of length 
    \begin{align*}
      \ell(P') + |(M_s' \symdiff E(P'_G))  \symdiff M_t|/2 + \delta_{P'} &= c(P') + |M_s' \symdiff M_t'|/2 + \delta_{P'}\\
      	& \leq c(P) + |M_s \symdiff M_t|/2 + \delta_P \enspace.
    \end{align*}
\end{proof}

\begin{claim}
    Let $P^*$ be a minimum-cost $XY$-path in $D$. Then a transformation
    from $M_s$ to $M_t$ via a matching $M$ that is not inclusion-wise maximal
    has length at least $c(P^*) + |M_s \symdiff M_t|/2 + \delta_{P^*}$.
    \label{claim:fpt:lb}
\end{claim}
\begin{proof}[Proof of Claim~\ref{claim:fpt:lb}]
    Let $N^*$ be a matching that is not inclusion-wise maximal, such that
    there is a shortest transformation of length $L$ from $M_s$ to
    $M_t$ via $N^*$. So there is an edge $e \in E(G)$, such that $N^* + e$
    is a matching. Since $N^* + e$ contains more edges than $M_s$, the
    symmetric difference of $N^*+e$ and $M_s$ contains an augmenting path,
    which corresponds to an $xy$-path $Q^{*}$ in $D$, where $x \in X$ and $y \in Y$. Then
    \begin{equation}
        L \geq \ell(Q^*) + |(M_s \symdiff E(Q^*_G)) \symdiff M_t| / 2 + \delta_{Q^*_G} = c(Q^*) + |M_s \symdiff M_t| / 2 + \delta_{Q^*_G}\enspace.
        \label{eq:fpt:lb}
    \end{equation}
    The first inequality follows, since the shortest transformation performs at
    least $\ell(Q^*)$ exchanges to reach the matching $N^*$. By
    Lemma~\ref{lemma:ptime:msr}, a shortest transformation from $M$ to $M_t$
    has length $|M \symdiff M_t|/2 + \delta_{Q^*} = |(M_s \symdiff E(Q^*_G))
    \symdiff M_t| + \delta_{Q^*}$. 
    
    Let us prove the final equality of~\eqref{eq:fpt:lb}. The argument is
    similar to the proof of Claim~\ref{claim:fpt:ub}. For an arc $a$ of $D$,
    let $u_a, v_a, w_a$ be the path in $G$ that corresponds to $a$. Note that
    $Q^*$ visits an
    $(M_s, M_t)$-alternating path that contains a vertex in $Y$ at most once;
    otherwise there is a shorter transformation. If $Q^*$ visits such a path
    $F$, we split $Q^*$ into two paths $R$ and $R'$ at a vertex $z$, such that
    $z$ is the last vertex of $Q^*$ not on $F$. If $Q^*$ visits no such path,
    then we let $R'$ be empty and $R := Q^*$.
    For each arc $a$ of $R'$ in reverse order (from
    $y$ to $z$), we do the following. The vertex $w_a$ is
    matched to the vertex $v_a$ by $M_s$. Furthermore, $w_a$ must be matched to a
    vertex $w'$ by $M_t$. Note that if $a$ is the first arc of $R'$ we process,
    then $w_a$ must be covered by $M_t$, otherwise $M_t$ is not maximal. In each step, we
    exchange $v_aw_a$ for $w_aw'$, which decreases the size of the symmetric
    difference of the current matching and $M_t$ by two. Let $M$ be the
    matching obtained be these exchanges. Then $|M \symdiff M_t| = |M_s
    \symdiff M_t| - 2\ell(R')$. Now, we perform an exchange for each arc of $R$
    in forward direction. Note that an arc of cost two corresponds to an
    exchange that increases the size of the symmetric difference by two.
    Similarly, an arc of cost one (zero) corresponds to an exchange that does
    not change the size of the symmetric difference (decreases the size of the
    symmetric difference by two). By keeping track of the size of the symmetric
    difference with respect to $M_t$, the final equality of~\eqref{eq:fpt:lb}
    follows.

    Recall that $P^*$ is a shortest $XY$-path in $D$. By Claim~\ref{claim:fpt:ub},
    there is a transformation from $M_s$ to $M_t$ of length at most $c(P^*) +
    |M_s \symdiff M_t|/2 + \delta(P^*_G)$. By the optimality of $P^*_G$ and
    observing that visiting a cycle $C \in \mathcal{C}$ incurs a cost of at
    least one for an $XY$-path, we obtain
    \begin{equation}
        c(Q^*) + |M_s \symdiff M_t| / 2 + \delta_{Q^*_G} \geq c(P^*) + |M_s \symdiff M_t| / 2 + \delta_{P^*_G}\enspace.
        \label{eq:fpt:lb2}
    \end{equation}
    We combine~\eqref{eq:fpt:lb} and~\eqref{eq:fpt:lb2} to prove the claim.
\end{proof}

Note that we can find in polynomial-time a minimum-cost $XY$-path $P^*$ in
$D$, for example by Dijkstra's algorithm. Hence, it suffices to transform
$M_s$ into the matching $M_s \symdiff E(P^*_G)$, which is not
inclusion-wise maximal, and then use the algorithm from
Lemma~\ref{lemma:ptime:msr} to transform the resulting matching into $M_t$.
By Claim~\ref{claim:fpt:lb}, the transformation has minimal length with
respect to all transformations from $M_s$ to $M_t$ via a matching that is
not inclusion-wise maximal.

\subsection{Remaining Proofs omitted from Section~\ref{sec:FPT}}

\begin{proof}[Proof of Fact~\ref{fact:fpt:oneside}]
     If an edge does not occur in some maximum matching of the input graph $G$
     then it is useless for reconfiguration. Therefore, we may assume that
     every edge of $G$ occurs in some maximum matching.  Also, we may assume
     that $G$ is connected, otherwise we consider each component separately.
     Now consider the Edmonds-Gallai decomposition $D$, $A$, $C$ of the graph
     $G$. The odd components of the decomposition are just single vertices
     since they are factor-critical, and no factor-critical graph is bipartite.
     By our assumptions above, $A$ is an independent set and $C$ is empty (no
     edge between $A$ and $C$ occurs in a maximum matching; $G$ is connected).
     Therefore, $(D, A)$ is a bipartition of $V(G)$ with all exposed vertices on
     the same side, $D$. 
\end{proof}

\begin{proof}[Proof of Proposition~\ref{prop:dst:assumptions}]
	Observe that since $F$ is a directed tree, the indegree of each vertex of $F$
	except $r$ is precisely one. Let $A'$ be the arc-set of $F$.\smallskip
	
	\noindent Proof of~\ref{itm:dst:patharcs}. Let $P \in \mathcal{P}$. Suppose
	that $F$ is missing at least one arc of $A(P)$. Since the indegree of each
	vertex of $F$ except $r$ is precisely one (they are terminals), there must be two nodes $x, y \in
	V(P)$, such that $xy \in A(P) \setminus A'$ and the predecessor $x'$ of $y$
	is not in $V(P)$. Then $F - x'y + xy$ connects $r$ to each terminal at cost
	$c(F - x'y + xy) < c(F)$ since we replace an arc of weight two by an arc of
	weigt one. Indeed, as we observed, only one arc entering in $y$ is special.
	This contradicts the optimality of $F$. \smallskip
	
	\noindent Proof of~\ref{itm:dst:cyclearcs}. Let $C \in \mathcal{C}$. If $A'$
	contains each arc of $A(C)$, then $F$ is not a tree.  So $F$ misses at least
	one arc. Suppose that $F$ is missing at least two arcs of $A(C)$. Since the
	indegree of each vertex of $F$ except $r$ is precisely one, there must be two
	vertices $x, y \in V(C)$ as follows. The arc $xy$ is in $A(C)$, but not in
	$A'$ and the predecessor of $y$ is not in $V(C)$. Let $a$ be the in-arc of
	$y$ in $F$. Then $F - a + xy$ connects $r$ to each terminal at cost $c(F - a
	+ xy) < c(F)$. This contradicts the optimality of $F$.  \smallskip
	
	\noindent Proof of~\ref{itm:dst:pathroot}. Let $P \in \mathcal{P}$. Let $v
	\in V(P)$ be the $M_s$-free vertex of $P$. By~\ref{itm:dst:patharcs}, the
	tree $F$ contains the arcs $A(P)$ and each vertex of $F$ except $r$ has
	indegree precisely one. Therefore, there is only a single in-arc $pv$, where
	$v \in V(P)$ and $p \notin V(P)$. If $p \neq r$, then we may replace $pv$ by
	$rv$ and obtain a directed tree $F - pv + rv$ that connects $r$ to each
	terminal at cost $c(F - pv + rv) < c(F)$. Again, this is a contradiction to
	the optimality of $F$.
\end{proof}

\begin{proof}[Proof of Claim~\ref{claim:indep} of Lemma~\ref{lemma:dst:sequence}]
  The three first statements follow directly from the definition of the
  transformation~\eqref{eq:Mi}.

  Let us prove the fourth point. First note that $a_C$ exists and
  all the arcs of $A(C)$ but $a_C$ are in $F$ by Proposition~\ref{prop:dst:assumptions}. 
  Since $a$ has been visited backwards, it implies that $a'$ has been visited
  forwards (since $a'$ is an arc of the unique path from $r$ to $a$).

  $a_j \in A'$ by Proposition~\ref{prop:dst:assumptions}.  We first
  show that the statement holds in step $j$.  If $a_{j-1}=a_j = u'u$
  then $u$ is $M_{j-1}$-free (by the second rule of~\eqref{eq:Mi})
  and $M_j = M_{j-1}$ (by the third rule of~\eqref{eq:Mi}).  So we
  can assume that $a_{j-1}$ is visited backwards. Since $a_C$ is not
  in $F$ and $a_C$ is the only out-arc of $u$ which is special, the
  arc $a_{j-1} = uv$ is neither special nor artificial. 
  By definition of rules four and five of~\eqref{eq:Mi}, we have that
  $\bar a_{j-1}=vu \in M_{j-2}$ is moved to an edge not incident to
  $u$ in $M_{j-1}$.  Since $M_{j}=M_{j-1}$, we have that $M_j$ is
  $u$-free.

  So we can assume that $a_j \ne a$. Let $j'$ be the step where $a$ is visited
  backwards.
  Since $F$ is a tree and we consider a DFS traversal, no arc incident
  to $w$ is considered between steps $j'$ and $j$ since $a'$ has not 
  been visited upwards. So if an edge $e$
  incident to $u$ is added in $M_j$, it
  can only be by the fourth rule of~\eqref{eq:Mi}. So $u$ is incident
  to an edge $e_{C'}$ for $C' \in \mathcal{C}$. Since $a'$ has not been visited
  backwards, neither is the arc entering in $C$. So $C \ne C'$, 
  a contradiction since it would mean that $u$
  belongs to both cycles $C$ and $C'$, that must be disjoint. 
\end{proof}

\setcounter{claim}{1}
\begin{claim}
  For $0 \leq i \leq m$, the set $M_i \subseteq E$ is a matching of $G$.
  \label{claim:dst:matchingsequence}
\end{claim}
\begin{proof}[Proof of Claim~\ref{claim:dst:matchingsequence} of Lemma~\ref{lemma:dst:sequence}]
  The statement holds for $i=0$, so let us assume that for some fixed
  $i < m$ and $1 \leq j \leq i$, we have that $M_j$ is a
  matching of $G$. We show that so is
  $M_{i+1}$. Note that $|M_{i+1}| = |M_i|$ by the definition
  of~\eqref{eq:Mi}. Let us prove that $M_{i+1}$ is 
  a matching. If $M_{i+1} = M_i$ then the conclusion indeed
  holds, so we assume that $M_{i+1} \neq M_i$. We distinguish two
  cases depending on whether the arc $a_{i+1} = u_{i+1}w_{i+1}$ is
  visited for the first or second time.

  \setcounter{case}{0}
  \begin{case}
    The arc $a_{i+1}$ is visited for the first time (i.e., $a_{i+1}$ is traversed downwards).
  \end{case}
  We can assume that $a_{i+1}$ is not artificial, since otherwise $M_i =
  M_{i+1}$. Let $a_i=u_iw_i$ be the arc visited before $a_{i+1}$.  We
  distinguish several subcases depending on which rule of~\eqref{eq:Mi} was
  applied to obtain $M_i$ from $M_{i-1}$. 
  If the first rule of~\eqref{eq:Mi} was applied to obtain $M_i$, the
  arc $a_{i}$ is artificial. Then $w_i=u_{i+1}$ is $M_{i-1}$-free by
  Claim~\ref{claim:indep}. Since $M_{i-1}=M_i$ by the first rule of~\eqref{eq:Mi},
  $u_{i+1}$ also is $M_i$-free.
  So $M_{i+1}$ is a matching of $G$. 
  If the second rule of~\eqref{eq:Mi} was applied, then we have that
  $u_{i+1}=w_i$ and $M_i = M_{i-1} - e_i + \bar e_i$, so $w_i$ is
  $M_i$-free as required.
  So the third, fourth or fifth rule of~\eqref{eq:Mi} was applied, 
  and then $a_{i+1}$ is traversed downwards and $a_i$ is traversed upwards. 
  Since arcs of larger weight are traversed first and at most one outgoing arc of
  $u_i$ has weight one, we have that $a_i$ has weight two. Therefore,
  the arc $a_i$ cannot be special and rule three of~\eqref{eq:Mi}
  is not applicable,
  Finally, if the fourth or fifth rule of~\eqref{eq:Mi} were applied, then 
  $u_i=u_{i+1}$ and, by definition of $M_i$, in both cases 
  the vertex $u_i$ is $M_i$-free.
  Therefore, the set $M_{i+1} = M_i - e_i + \bar e_i$ is a matching
  of $G$.

  \begin{case}
    The arc $a_{i+1}$ is visited for the second time (i.e., $a_{i+1}$ is traversed upwards).
  \end{case}
  If $M_{i+1}$ is obtained via the first or the third rule, then
  $M_{i+1}=M_i$ and the conclusion holds. Furthermore, rule two is
  not applicable, since $a_{i+1}$ is traversed upwards by assumption.
  So $M_{i+1}$ is obtained by rules four or five of
  of~\eqref{eq:Mi}. As in Case 1, we distinguish several subcases,
  depending on which rule of~\eqref{eq:Mi} was applied to obtain
  $M_i$ from $M_{i-1}$ when processing $a_{i} = u_iw_i$.

  If the first rule of~\eqref{eq:Mi} was applied, then $u_i = w_{i+1}
  = r$, which contradicts our construction of the instance.  If rule
  two of~\eqref{eq:Mi} was applied, then $M_i = M_{i-1} - e_i + \bar
  e_i$ and $a_{i+1}=a_i$. So in particular, the fourth rule
  of~\eqref{eq:Mi} cannot be used to obtain $M_{i+1}$. Indeed $a_i=a_{i+1}$
  cannot be both, special ($a_i$ is because of rule four for $a_{i+1}$) 
  and not special ($a_{i+1}$ is not since rule three is not applied).
  So we have either $M_{i+1}
  = M_{i-1}$ or $M_{i+1} = M_i$ and the conclusion holds. We may
  therefore assume that $a_i$ is traversed for the second time and
  hence, we have $u_i = w_{i+1}$. So $M_{i+1}$ is the result of rules
  four or five of~\eqref{eq:Mi} and $M_{i}$ from rules three, four,
  or five of~\eqref{eq:Mi}.

  Assume first, that the fifth rule of~\eqref{eq:Mi} was applied to
  obtain $M_{i+1}$. We first prove by contradiction that $M_i$ 
  cannot be obtained using rule three of~\eqref{eq:Mi}. 
  Indeed otherwise, $a_i$ would be special and belong to $Z$ and
  $a_{i+1}$ would not be special. Since rule four is not applied to obtain
  $M_{i+1}$ and $a_{i+1}$ would be the only arc entering $V(Z)$. By
  Proposition~\ref{prop:dst:assumptions} $Z$ would be a path.  But then
  $u_i$ would incident to $r$, so $a_{i+1}=ru$ by
  Proposition~\ref{prop:dst:assumptions}, a contradiction.
  If rules four or five of~\eqref{eq:Mi} were applied to obtain $M_i$, then
  we have that $u_{i} = w_{i+1}$ is $M_i$-free, so $M_{i+1} = M_i - \bar
  e_{i+1} + e_{i+1}$ is a matching of $G$. 

  Finally assume that the fourth rule of~\eqref{eq:Mi} was applied
  to obtain $M_{i+1}$.  By definition of rule four, rule three
  of~\eqref{eq:Mi} was applied to $a_i$. So $M_{i+1} = M_{i} - \bar
  e_{i+1} + e_C$ for some $C \in \mathcal{C}$. Note that by
  definition of $a_C$, we have that $w_{i+1}=u_i$ is a vertex of $C$.
  Let us denote by $u$ the other endpoint of $a_C$. By
  Proposition~\ref{prop:dst:assumptions}, all the arcs of $C$ but
  $a_C$ are in $F$ and since $F$ is traversed in DFS order, they were
  already traversed backwards. 
  The last point of Claim~\ref{claim:indep} ensures that $u$ is
  $M_{i-1}$-free. Since $M_i = M_{i-1}$, it is also $M_i$-free. Thus
  the exchange of $\bar e_{i+1}$ for $e_C$ results in a matching.

  This concludes the proof of Claim~\ref{claim:dst:matchingsequence}.
\end{proof}

\begin{proof}[Proof of Claim~\ref{claim:dst:matchingtarget} of Lemma~\ref{lemma:dst:sequence}]
  By the definition of~\eqref{eq:Mi}, an edge $vw \in M_s$ is exchanged
  at most twice during the algorithm, once when the arc $a$ of $F$
  ending in $w$ is traversed downwards and once when it is traversed
  backwards.

  Let us first prove that at the end of the algorithm, if an edge $e =
  vw \in M_s$ is exchanged, then it is exchanged with an edge in $M_t$.
  Let $a=uw$ be the unique in-arc of $w$ in $F$. When $a$ is traversed
  forwards, the edge $vw$ is exchanged with $uv$ (rule two
  of~\eqref{eq:Mi}). If $a$ is special, then $uv$ is in $M_t$ (by the
  definition of special edge) and $uv$ is not exchanged when $a$ is
  visited backwards (rule three of~\eqref{eq:Mi}), so the conclusion
  holds in this case. Hence we can assume that $a$ is not special. If
  $vw \notin M_s \symdiff M_t$, then $w$ is not in a cycle $C$ of
  $\mathcal{C}$, and then it has no out-arc that is special. Thus, when
  $a$ is traversed backwards, rule five of~\eqref{eq:Mi} applies, and
  then $vw \in M_t$. 
  Let us finally assume that $a$ is not special but $vw \in M_t - M_s$.
  By Proposition~\ref{prop:dst:assumptions}, the arc $a=uw$ is the
  unique arc entering in some item $Z \in \mathcal{C} \cup
  \mathcal{P}$. Since $a$ is not artificial,
  Proposition~\ref{prop:dst:assumptions} ensures that $Z=C \in
  \mathcal{C}$.  By construction, when $a$ is traversed backwards, the
  edge $uv$ is exchanged for $e_C$ which is, by definition, in $M_t$.

  To conclude, we simply have to prove that each edge $M_s \setminus
  M_t$ has been exchanged at some point. This is indeed the case since,
  for each arc $vw \in M_s - M_t$, the vertex $w$ is a terminal. Since
  $F$ is a Steiner tree, it connects $r$ to each terminal.
\end{proof}

\begin{proof}[Proof of Lemma~\ref{lemma:dst:connectivity}]
	Let $w \in T$ be a terminal. Then, by the definition of the arc-set $A$, there is some edge
	$vw \in M_s$. By the definition of the Edmonds-Gallai decomposition in
	Section~\ref{sec:diameter}, there is an even-length alternating path from an
	$M_s$-free vertex to $w$.  Therefore, $w \in D(G)$.  Let $vw$ be the
	$M_s$-edge incident to $w$. Then $v \in A(G)$. Since, by
	Theorem~\ref{thm:egd}, any maximum matching of $G$ induces an $A(G)$-perfect
	matching into $D(G)$ and $C(G)$ is completely matched by any maximum
	matching, it is only possible to exchange $vw$ for some edge $vz$, where $z$
	is a neighbor of $v$. Since $vw \notin M_t$, there is some step $i$, such
	that $M_i = M_{i-1} - vw + vz$, for some neighbor $z$ of $v$.  Therefore, the
	indegree of $w$ in $D_{\mathcal{S}}$ is at least one.  If $z$ is
	$M_s$-free, then there is some $rw$-path in $D_{\mathcal{S}}$ and we are
	done. Suppose this is not the case, so $M_{i}$ does not cover $z$, but $M_s$
	does. Then, there is some index $i' < i$, such that $M_{i'}$ does not cover
	$z$, but $M_{i'-1}$ does. Among all such indices, let $i'$ be the minimal one.
	So in particular the edge of $M_{i'-1}$ incident to $z$ is the same as in $M_s$.
	That is, there is some neighbor $v'$ of $z$ and
	some neighbor $z'$ of $v'$, such that $M_{i'} = M_{i'-1} - zv' + v'z'$. Since
	$z \in D(G)$, we have that the indegree of $z$ in $D_{\mathcal{S}}$ is at
	least one. Again, if $z$ is $M_s$-free then we are done. We can repeat this
	argument until we reach either an $M_s$-free vertex or either $M_s$ or $M_t$
	are not maximum, or to $\mathcal{S}$ is not a reconfiguration sequence
	transforming $M_s$ to $M_t$.  We conclude that there is an $rw$-path in
	$D_\mathcal{S}$ for each $w \in T$.  
\end{proof} 

\begin{proof}[Proof of Lemma~\ref{lemma:fpt:correctness}]
    We show that we can decompose $\tau$ into two transformations $\tau_1$ and
    $\tau_2$, such that $\tau_1$ is a transformation from $M_s$ to $M_U(S)$ for
    some $S \in \{U, W\}^{\mathcal{P}\cup\mathcal{C}}$ and $\tau_2$ transforms
    $M_U(S)$ into $M_t$. Our claim then follows, since $\alpha(S^*)$ is a
    shortest such transformation.

    We consider $\tau$ as a sequence $((e_i, f_i))_{1 \leq i \leq m}$ of
    exchanges, where $e_i = v_iw_i$ and $f_i = u_iv_i$. For each $1 \leq i \leq
    m$, we have that $M_i = M_{i-1} - e_i + f_i$. Note that for each exchange
    $(e_i, f_i)$, the vertex $v_i$ is $M_i$-free. Consider the transformations
    $\tau_1 := ((e_i, f_i))_{1 \leq i \leq m :\, u_i \in U}$ and $\tau_2 :=
    ((e_i, f_i))_{1 \leq i \leq m :\, u_i \in W}$.  We show that we can apply
    $\tau_1$ to $M_s$ to obtain an intermediate matching $M$ and that we can
    apply $\tau_2$ to $M$ to obtain $M_t$. For $j \in \{1, 2\}$, let
    $|\tau_j|$ be the length of $\tau_j$ and re-index the exchanges from $1$ to
    $|\tau_j|$.

    Suppose for a contradiction that for some index $i$, the exchange $(v_iw_i,
    u_iv_i)$ of $\tau_1$ cannot be performed and let $i$ be the smallest such
    index. 
    This means, that $v_iw_i$ is not in the current matching because some
    exchange $(v'_jw'_j, u'_jv'_j)$ of $\tau_2$ needs to be performed before
    $(v_iw_i, u_iv_i)$. Again, we assume that $j$ is the smallest such index.
    We distinguish two cases. 
    If $(v_iw_i, u_iv_i)$ cannot be performed since $u_i$ is not exposed, then
    $w'_j = u_i$, so $u'_j \in U$, which contradicts the construction of
    $\tau_2$. 
    On the other hand, if the exchange $(v_iw_i, u_iv_i)$ cannot be performed
    since $v_iw_i$ is not in the current matching, then $u'_jv'_j = v_iw_i$.
    Then we can perform exchanges $1, 2, \ldots, j-1$ of $\tau_2$ and $1, 2,
    \ldots, i-1$ of $\tau_1$ to obtain a matching $M'$. But $u_iv_i$ is an edge
    of $G$ and $u_i$ and $v_i$ are both $M'$-free. Therefore, $M'$ is not
    inclusion-wise maximal. Since $\tau$ contains all the exchanges we
    performed so far to obtain $M'$, there is a transformation from $M_s$ to
    $M_t$ via $M'$ that has length at most $m$ by Lemma~\ref{lemma:ptime:msr},
    a contradiction.

    Let us now prove that $\tau_2$ is a transformation from $M'$ to $M_t$.
    The argument is similar to the one above.
    Suppose for a contradiction that for some index $i$, the exchange $(v_iw_i,
    u_iv_i)$ of $\tau_2$ cannot be performed and let $i$ be the smallest such
    index.
    If $(v_iw_i, u_iv_i)$ cannot be performed since $u_i$ is not exposed, then
    $w'_j = u_i$, so $u'_j \in V$, which contradicts the construction of
    $\tau_1$.
    On the other hand, if the exchange $(v_iw_i, u_iv_i)$ cannot be performed,
    because $v_iw_i$ is not in the current matching, then $v'_jw'_j = v_iw_i$.
    Then we can perform exchanges $1, 2, \ldots, j$ of $\tau_1$ and $1, 2,
    \ldots, i-1$ of $\tau_2$ to obtain a matching $M'$.
    But $u_iv_i$ is an edge of $G$ and $u_i$ and $v_i$ are both $M'$-free.
    Therefore $M'$ is not inclusion-wise maximal.
    Since $\tau$ contains all the exchanges we performed so far to obtain $M'$,
    there is a transformation from $M_s$ to $M_t$ via $M'$ that has length at
    most $m$ by Lemma~\ref{lemma:ptime:msr}, a contradiction.

    We show that if $(e_i, f_i)$ is the last exchange of $\tau_1$ that
    involves the edge $f_i$, then $f_i$ cannot be moved by $\tau_2$, so $f_i$ is in the
    target matching $M_t$. 
    Suppose for a contradiction, that there is some
    index $j$, such that
    the $(e_j, f_j)$ of  $\tau_2$ has  $f_i = e_j$. Let $j$ be the smallest
    such index. Then we can apply $\tau_1$ to $M_s$ and the perform all
    exchanges of $\tau_2$ up to index $j$. Let $M'$ be the resulting matching.
    Then $M'$ is not inclusion-wise maximal, since $M' + e_i$ is a matching of
    $G$. Therefore, by Lemma~\ref{lemma:ptime:msr}, there is a transformation
    from $M_s$ to $M_t$ via $M'$ of length at most $m$, a contradiction.

    By swapping $\tau_1$ and $\tau_2$ and using an analogous argument, we may
    conclude that if $(e_i, f_i)$ is the last exchange of $\tau_2$ that
    involves $f_i$, then $f_i$ is in $M_t$.
    Therefore, each exchange $(e, f)$ of $\tau$, such $f$ is involved for the
    last time occurs as a final exchange involving $f$ either in $\tau_1$ or
    $\tau_2$. Therefore, we obtain $M_t$ by applying $\tau_1$ followed by
    $\tau_2$ to $M_s$.

    Let $M$ be the matching obtained by  applying $\tau_1$ to $M_s$ and
    consider the set $\mathcal{P}$ of paths and the set $\mathcal{C}$ of cycles
    in $M_s \symdiff M_t$. We claim that for each $F \in \mathcal{C} \cup
    \mathcal{P}$, when reaching $M$, we have either completely reconfigured $F$
    or performed no change at all on $F$. That is, for each $F \in \mathcal{C}
    \cup \mathcal{P}$, we have that $M \cap E(F)$ is either $M_s \cap E(F)$ or
    $M_t \cap E(F)$. Suppose for a contradiction, that this is not the case for
    some $F \in \mathcal{C} \cup \mathcal{P}$. If some edge $f \notin M_s \cup
    M_t$ is incident to $F$ in $M$, then, by the argument above, the edge $f$
    cannot be involved in any exchange of $\tau_2$ and will remain in the final
    matching, a contradiction. 
    If the matching $M$ contains at least one edge in  $M_s \cap E(F)$ and at
    least one edge in $M_t \cap E(F)$, then there is an $M$-free vertex $u \in
    V(F)$. 
    Due to the construction of $\tau_1$, we
    have that $u \in U$. But each exchange in $\tau_2$ only swaps two edges
    incident to a common vertex in $U$, so $u$ is $M_t$-free. Therefore, $F$
    must be a path and $u$ is one of its end vertices. Since all edges of $M$
    incident to $F$ are in $M_s \cup M_t$ and $u$ is $M_t$-free, we have
    that $M \cap E(F) = M \cap E(M_t)$, a contradiction.

    Finally, we choose $S \in \{U, W\}^{\mathcal{C}\cup\mathcal{P}}$, such that
    for $F \in \mathcal{C} \cup \mathcal{P}$, we have that $S_F = U$ if $M_t
    \cap E(F)  = M \cap E(F)$ and $S_F = W$ otherwise. By the definition of
    $\tau_1$ and our arguments above, we have that $M = M_U(S)$. By the
    optimality of the choice of $S^*$, we have that $|\alpha(S^*)| \leq
    |\alpha(S)| = m$.
\end{proof}

%% file: appendix3.tex
\section{Proofs Omitted from Section~\ref{sec:diameter}}

\begin{proof}[Proof of Theorem~\ref{thm:connected}]
	Let $\nu(G)$ be the size of a maximum matching in $G$ and let $A := A(G)$, $D
	:= D(G)$ and $C := C(G)$ the partition of $V(G)$ according to the
	Edmonds-Gallai decomposition.
	If $k < \nu(G)$ then
	$\matchgraph(G)$ is connected according to the proof of~\cite[Proposition
	1]{Ito:11}, so let $k = \nu(G)$.  From Theorem~\ref{thm:egd}, we have that
	$G[C]$ is perfectly matched by any maximum matching of $G$. Let us prove that
	$\matchgraph(G)$ is connected if and only if $G[C]$ admits a unique perfect matching.
	
	First consider the case that $G[C]$ admits two distinct perfect matchings. Let $M_1$ and $M_2$
	be extensions of the two distinct perfect matchings on $G[C]$ to maximum
	matchings of $G$. Assume by contradiction that there is a transformation and let $M$ be
	the matching just before the first modification of an edge in $G[C]$. By definition of $C$, 
	there is no $M$-alternating path from an $M$-free vertex to $C$ for any maximum matching
	$M$ of $G$. In particular, there is no exchange $M-e+f$, where $e \in E(G[C])
	\cap M$ and $f \in E(G) \setminus M$ for any maximum matching $M$ of $G$
	(since otherwise a maximum matching would not be perfect on $G[C]$). And since $M$ is $C$-perfect
	one cannot replace an edge in $G[C]$ by another.
	Therefore, all the matchings in the connected component of $M_1$ in $\matchgraph(G)$
	agree on $G[C]$. Thus $M_1$ and $M_2$ cannot be connected in $\matchgraph(G)$, and then
	$\matchgraph(G)$ is not connected.  
	
	On the other hand, suppose that all the maximum matchings agree on $G[C]$.  
	Let $M_1$ and $M_2$ be two maximum   matchings of $G$ and suppose that $G[M_1 \symdiff M_2]$ 
	contains an even-length $(M_1, M_2)$-alternating cycle $C$.  Then $C$ must be contained
	in $G[D \cup A]$. By the definition of $D$ and $A$ in Theorem~\ref{thm:egd},
	there is an $M_1$-alternating path from an $M_1$-free vertex to $C$.
	Therefore, by Lemma~\ref{lemma:cyclereconf}, $(G, M_1, M_2)$ is a \yes
	instances of \MR.  It follows that if $G[C]$ has a unique perfect matching,
	then $\matchgraph(G)$ is connected, so $\diam(\matchgraph(G))$ is finite.
\end{proof}

\begin{proof}[Proof of Lemma~\ref{lemma:dist:cover}]
	Let $e \in E$. Reconfiguring $M \cap E(C_e)$ to $N \cap E(C_e)$ requires that
	$p_v^1$ is exposed for some $v \in V$ such that $e$ is incident to $v$. If
	this is the case, then the reconfiguration requires precisely three
	exchanges. So we need $3|E|$ exchanges to reconfigure the cycles on four vertices.
	
	To reconfigure a cycle, we need that one of its neighbors is exposed. 
	By definition of $\tau(H)$ and by construction at least $\tau(H)$ vertices have to 
	be exposed at some step of the algorithm. 
	Let us prove that two steps are needed to expose a vertex $p_u^1$ if $p_v^1$ 
	is exposed. Since all but one vertex are covered by the matching, we have perfect matchings on
	all the $C_4$. Thus to expose $p_u^1$ we need to push the edge $p_u^1p_u^2$ on $p_u^2t$. Since
	there is an edge incident to $t$ (otherwise two vertices are exposed), we need to push this edge.
	Thus at least two steps are needed to expose $p_u^1$. Finally we need one step to expose one vertex
	$p_u^1$ at the beginning and one step to expose $t$ at the end.
	Therefore, $\dist(M, N) \geq 3|E(H)| + 2\tau(H)$. And one can easily prove that a transformation of 
	this length exists.
\end{proof}

\begin{proof}[Proof of Lemma~\ref{lemma:diam:ub}]
	Let $M^*$ and $N^*$ be two maximum matchings of $G$ of maximal distance in
	$\matchgraph(G)$. Note that due to the maximality of $\dist(M, N)$, for each $e \in
	E$ such that $M^*$ and $N^*$ leave no vertex of $C_e$ exposed, we have that
	either $M^* \cap E(C_e) = \{c_e^1c_e^2, c_e^3c_e^4\}$ and $N^* \cap C_e =
	\{c_e^2c_e^3, c_e^4c_e^1\}$ or vice versa (indeed a transformation where they are different
	can immediately be adapted into a transformation where they are the same). 
	So we may assume that $M^*$ agrees with
	$M$ and $N^*$ agrees with $N$ on each $e \in E$ such that $C_e$ contains no
	exposed vertex of $M^*$ or $N^*$. Now, due to the construction of $G$ and
	since $M$ and $N$ leave $t$ exposed, $\dist(M^*, M) \leq 3$ and
	$\dist(N^*, N) \leq 3$. Indeed if some $p_u^1$ or $q_i$ is exposed, we can expose $t$ in at most $3$ steps.
	Similarly, if a vertex of a $C_4$ is exposed, then we can expose $t$ in two steps.
	
	Hence we have
	\begin{equation}
	\diam(\matchgraph(G)) = \dist(M^*, N^*) \leq \dist(M, N) + 6
	\label{eq:dist:ub}
	\end{equation}
	Now let $M'$ ($N'$) be a maximum matching of $G$ that leaves $q_6$ exposed
	and agrees with $M$ ($N$) on each $C_e$ for each $e \in E$. Then, any
	reconfiguration sequence from $M'$ to $N'$ must include $M$ and $N$ and
	since $\dist(M', M) = \dist(N', N) = 3$, we have that $\dist(M', N') =
	\dist(M, N) + 6$. By~\eqref{eq:dist:ub}, the matchings $M'$ and $N'$ have
	maximal distance in $\matchgraph(G)$ and the lemma follows.
\end{proof}

\begin{proof}[Proof of Theorem~\ref{thm:dpc}]
        Note that the length of a shortest transformation between two matchings
        of $G$ is bounded by $O(|V^2|)$~\cite{Ito:11}.  Due to the polynomial
        size of a certificate, the question ``is the distance of two matchings
        in $\matchgraph(G)$ at most $\ell$`` is an \NP-question. Therefore,
        \EMDIST is in $\D^\P$. In order to show that \EMDIST and \EMDIAM are
        $\D^\P$-hard, we give a polynomial-time reduction from
        \textsc{Sat-Unsat}, which is complete for $\D^\P$~\cite{PY:84} and
        defined as follows.
	
	\begin{quote}
	  \textsc{Sat-Unsat} \\
	  \textbf{Input:} \textsc{3-CNF} formulas $F$ and $F'$\\
	  \textbf{Output:} \yes if and only if $F$ is satisfiable and $F'$ is not.
	\end{quote}
	
	Note that we may as well reduce from \textsc{Exact Vertex Cover}, which
	is also known to be $\D^\P$-complete~\cite[Theorem 5.4]{Cai:88}.
	However, to make the proof more self-contained, we include the simple
	reduction step from \textsc{Sat-Unsat} to \textsc{Exact Vertex Cover}.
	Let $F$ be a 3-CNF formula with $m$ clauses and $n$ variables. 
	Recall that $H$ has a $\ell$-clique if and only if
	$\overline{H}$ a $\ell$-stable set if and only if $\overline{H}$ has a vertex
	cover of size $|V|-\ell$.  Using the standard reduction from \textsc{3-SAT}
	to \textsc{Clique} and adding a clique of size $m-1$ to the resulting graph
	and we have
	\begin{equation}
	\begin{aligned}
	F \in \textsc{3-SAT}& \Leftrightarrow \text{$H$ has a maximum clique of size $m$}  \\
	& \Leftrightarrow \text{$\overline H$ has a minimum vertex cover of size $|V|-m$}\\
	& \stackrel{Lemma~\ref{lemma:dist:cover}}{\Leftrightarrow} \dist_{\matchgraph(G)}(M, N) = 3|E(\overline{H})| + 2(|V(\overline{H})|-m) \\
	& \stackrel{Lemma~\ref{lemma:diam:ub}}{\Leftrightarrow} \diam_{\matchgraph(G)} = 3|\overline{H})| + 2(|V(\overline{H})| - m + 3)
	\end{aligned}
	\label{eq:app:satdist}
	\end{equation}
	On the other hand, using similar arguments we have
	\begin{equation}
	\begin{aligned}
	  F \notin \textsc{3-SAT} & \Leftrightarrow \dist_{\matchgraph(G)}(M, N) = 3|E(\overline{H})| + 2(|V(\overline{H}|-m + 1) \\
	    & \Leftrightarrow \diam_{\matchgraph(G)} = 3|E(\overline{H})| + 2(|V(\overline{H}| -m + 4)
	\end{aligned}
	\label{eq:app:unsatdist}
	\end{equation}
	
	Let $(F_1, F_2)$ be an instance of \textsc{Sat-Unsat}. Let $H_1$ $(H_2)$ be
	the graph constructed from $F_1$ ($F_2$) in the reduction from \textsc{Sat}
	to \textsc{Clique}. Furthermore, let $G_1$, $M_1$, and $N_1$ ($G_2$, $M_2$ and $N_2$) be the graph and the two matchings obtained from $\overline{H_1}$ 
	($\overline{H_2}$) according to the construction given above. We may obtain
	in polynomial time from $(F_1, F_2)$ an instance $(G, M, N, k, \ell)$ of \EMDIST
	as follows: Let $G$ consist of two copies of $G_1$ and one copy of $G_2$ and
	let $M$ ($N$) be the obvious matching in $G$ created from two copies of $M_1$
	($M_2$) and one copy of $N_1$ ($N_2$).  Finally, let $k$ be the size of a maximum matching of $G$ and $\ell := 6|E(G_1)| +
	4(|V(G_1)| - m_1) + 3|E(G_2)| + 2(|V(G_2)| -m_2) + 2$.  Clearly, the instance
	$I$ can be constructed from $(F_1, F_2)$ in polynomial time.
	Using~\eqref{eq:app:satdist} and~\eqref{eq:app:unsatdist} it is readily verified that 
	\begin{align*}
	(F_1, F_2) \in \textsc{Sat-Unsat} \Leftrightarrow \dist_{\matchgraph(G)}(M, N) = \ell 
	\end{align*}
	and that
	\begin{align*}
	(F_1, F_2) \in \textsc{Sat-Unsat} \Leftrightarrow \diam_{\matchgraph(G)} = \ell+6
	\end{align*}
\end{proof}